\documentclass[submission,copyright,creativecommons]
{eptcs}
\providecommand{\event}{FICS 2024} 
\usepackage{url}    
\usepackage{breakurl}    

\usepackage{etoolbox}
\newcommand{\longv}{}

\usepackage{amsmath}
\usepackage{amsthm}
\usepackage{amsfonts}
\usepackage[T1]{fontenc}

\usepackage{coqdoc}

\usepackage{hyperref}

\usepackage{cleveref}

\usepackage{array}

\newcommand{\lname}[1]{\ensuremath{\mathsf{#1}}} 
\newcommand{\cln}[1]{\ensuremath{\mathsf{#1}^{\mathsf{cl}}}} 
\newcommand{\ipc}{\lname{IPC}}
\newcommand{\cpc}{\lname{CPC}}

\newcommand{\eql}[1]{\equiv_{#1}}

\newcommand{\deq}{:=}

\newcommand{\eqv}{\leftrightarrow}

\newcommand{\iln}[1]{\ensuremath{\mathsf{#1}^{\mathsf{int}}_\Box}}

\newtheorem{thm}{Theorem}
\newtheorem{theorem}[thm]{Theorem}
\newtheorem{corollary}[thm]{Corollary}
\newtheorem{lemma}[thm]{Lemma}

\theoremstyle{definition}
\newtheorem{probl}{Open Problem}
\newtheorem{remark}[thm]{Remark}{\bfseries}{\upshape}

\newcommand{\bro}{\textup{(}}
\newcommand{\brc}{\textup{)\,}}

\newcommand{\mycoqlinebreak}{\coqdoceol\coqdocindent{14.00em}}
\newcommand{\cdkw}[1]{\coqdockw{#1}}

\newcommand{\vrbl}[1]{\coqdocvariable{#1}}

\title{Ruitenburg's Theorem Mechanized and Contextualized}

\author{Tadeusz Litak\thanks{The results were obtained while the author was employed at FAU Erlangen-Nuremberg. In the final stages of preparing the print version, the author has been employed at University of Naples Federico II, supported by the PNRR MUR projects FAIR (No. PE0000013-FAIR).}
\institute{
University of Naples Federico II
}
\email{tadeusz.litak@unina.it
}}

\def\authorrunning{T. Litak}

\def\titlerunning{Ruitenburg's Theorem Mechanized and Contextualized} 

 \usepackage[font={small}]{caption, subfig}

\begin{document}

\maketitle

\begin{abstract}
In 1984, Wim Ruitenburg published a surprising result about periodic sequences in intuitionistic propositional calculus (IPC). The property established by Ruitenburg naturally generalizes local finiteness; recall that intuitionistic logic is not locally finite, even in a single variable. One of the two main goals of this note is to illustrate that most "natural" non-classical logics failing local finiteness also do not enjoy the periodic sequence property. IPC is quite unique in separating these properties. The other goal of this note is to present a Coq formalization of Ruitenburg's heavily syntactic proof. Apart from ensuring its correctness, the formalization allows extraction of a program providing a certified implementation of Ruitenburg's algorithm.
\end{abstract}

\bigskip 




\section{Introduction} \label{sec:intro}

In 1984, Wim Ruitenburg \cite{Ruitenburg84:jsl}  published  a surprising result about periodic sequences in intuitionistic propositional calculus (\ipc). For quite a while, the result seemed relatively neglected despite having been discussed, e.g., in Humberstone's monograph on logical connectives \cite{Humberstone2011}, and used in the work of the late Sergey Mardaev \cite{Mardaev1993:al,Mardaev1994:al,Mardaev1997:bsl,Mardaev1998:al,Mardaev2007:jancl}. Recently, however, there has been renewed interest \cite{Ghilardi2016,GhilardiS18,GhilardiGS20,GhilardiS20} (see \Cref{sec:related}), although as we are going see, Ruitenburg's discovery appears to deserve still broader attention. 

 Consider a propositional formula  $A$. Fix a propositional variable $p$, which can be thought of as representing the context hole or the  argument of $A$ taken as a polynomial (other propositional variables being additional constants). Given any other formula $B$, write $A(B)$ for the result of substituting $B$ for $p$. Also, write $A \eql{L} B$ for $\vdash_L A \eqv B$, where $L$ is a chosen system of propositional logic. Now define the natural iterated substitution operation
\ifdef{\longv}{
$$
A^0(p) := p, \quad A^{n+1}(p) := A(A^n(p)).
$$

}
{
$
A^0(p) := p, A^{n+1}(p) := A(A^n(p)).
$
}
 Such a sequence turns almost immediately into a cycle modulo $\eql{\cpc}$, i.e., the equivalence relation of Classical Propositional Calculus: 

\begin{lemma}[\cite{Ruitenburg84:jsl}, Lemma 1.1] \label{lem:cpc_ugpsp}
For any $A$, \quad $A(p)\eql{\cpc} A^3(p).$
\end{lemma}


The above observation can be reformulated as asserting that \cpc\ has \emph{uniformly globally periodic sequences} (ugps). A logic $L$  has this property if
 there exist $b$, $c > 0$ s.t. for any formula $A$, $A^b(p) \eql{L}A^{b+c}(p)$.  
However, ugps has still a rather strong logical form: two existential quantifiers preceding an universal one. Hence one can consider changing the order of quantifiers to weaken the property: 

\newcommand{\tbskip}{\hspace{0.7mm}}

\begin{center}
(eventually) periodic sequences: 

\medskip

\begin{tabular}{|@{\tbskip}l@{\tbskip}|@{\tbskip}>{$}c<{$}>{$}c<{$}>{$}c<{$}@{\tbskip}|@{\tbskip}>{$}c<{$}>{$}c<{$}>{$}c<{$}@{\tbskip}|}\hline
 & \multicolumn{3}{c@{\tbskip}|@{\tbskip}}{globally} & \multicolumn{3}{c@{\quad}|}{locally} \\ \hline
uniformly & \exists b. & \exists c>0. & \forall A. & \exists c>0.& \forall A. & \exists b.  \\ \hline
 parametrically  & \exists b. & \forall A. & \exists c>0. & \forall A. &  \exists b. & \exists c>0. \\ \hline
\end{tabular} \qquad  $A^b(p) \eql{L}A^{b+c}(p)$ 
\end{center}

 So, do standard non-classical propositional calculi, \ipc\ in particular, have at least \emph{plps} (\emph{parametrically locally periodic sequences})?\footnote{I am not aware whether terminology and distinctions used in the present paper have been systematically introduced before. The original work of Ruitenburg uses the term \emph{finite order}.} \ifdef{\longv}{ If not, is there something special about \cpc\ which makes results such as Lemma \ref{lem:cpc_ugpsp} possible?

One of many peculiarities of \cpc\ as seen from the general perspective of \emph{abstract algebraic logic} \cite{BlokP89:ams,FontJP03a:sl} is that it is \emph{finite}, i.e., determined by a single finite algebra of truth values. There are some other natural examples of finite logics  (mostly some fuzzy logics of finite chains, such as \L ukasiewicz three-valued logic), but in general, this property is rather rare among logics of importance in today's Computer Science. A somewhat more general property is \emph{local finiteness}: a logic is locally finite if given a finite set of propositional variables, one can form only finitely many non-equivalent formulas. The modal system \lname{S5}\ is a typical example of a logic which is locally finite without being finite.}{To begin with, we have an obvious observation:}

\begin{lemma} \label{lem:llf}
 Any locally finite logic has  plps.
\end{lemma}

\ifdef{\longv}{
\begin{proof}[Sketch]
Any sequence $p, A(p), A^2(p), A^3(p) \dots$ disproving the plps property would also disprove local finiteness. 
\end{proof}}{}

It is, however, well-known that \ipc\ is not locally finite: even in one propositional variable, there are infinitely many nonequivalent formulas. The exact description of the infinite algebra of formulas in one propositional variable is provided by the Rieger-Nishimura Theorem (see \cite{Rieger49,Nishimura60:jsl} and \cite[Ch. 7]{ChagrovZ97:ml} for references, also \Cref{sec:bounds_ens} herein).   As we are going to see in \Cref{sec:context}, most ``natural'' propositional logics which fail to be locally finite, fail to have (even parametrically locally) periodic sequences. 
 \ipc\ turns out to fare better. One can indeed show that (uniformly or parametrically) globally periodic sequences would be too much to expect, at least when formulas are allowed to contain other variables than $p$ itself \cite[\S 2]{Ruitenburg84:jsl}. But we do have
\begin{theorem}[\cite{Ruitenburg84:jsl}, Theorem 1.9]
\ipc\ has the ulps property:  for any $A$, there exists $b$ s.t. 
$A^b(p) \eql{\ipc} A^{b+2}(p)$. Moreover, $b$ is linear in the size of $A$.
\end{theorem}
In fact, Ruitenburg's theorem 
 is effective: the proof provides an algorithm to compute $b$ in question, as discussed in Sections \ref{sec:bounds_ens} and \ref{sec:bounds_lists} below. Moreover, as the periodic sequence property (in all its incarnations) transfers from sublogics to extensions in the same signature  (just like local finiteness and unlike uniform interpolation), we also get that all superintuitionistic logics (\emph{si-logics}) have ulps. This shows that unlike local finiteness, ulps does \emph{not}  guarantee the finite model property (\emph{fmp}), or even Kripke completeness. 
 
On the other hand, there is an obvious connection with fixpoint definability\footnote{In fact, the present author got interested in Ruitenburg's result for similar reasons, in the context of  ongoing joint work with Albert Visser on definability of fixpoint in intuitionistic modal logics involving (a strong version of) the L\"ob axiom. I would like to thank Albert Visser  for attracting my attention to Ruitenburg's work, for his support and comments on early drafts. Thanks are also due to Wim Ruitenburg for providing his recollections of how the theorem was proved. Furthermore, George Metcalfe kindly corrected my misunderstandings concerning the status of \lname{RM} \bro``\lname{R} with Mingle''\brc and provided me with several additional references. Finally, I would like to thank the referees of this and earlier incarnations of this paper, Alexis Saurin, Lutz Schr\"oder and the participants of Oberseminar of our group 
 at FAU  for discussions and suggestions. 
 }: If $A(p)$ is a monotone formula, then  $\{A^n(p)\}_{n \in \omega}$ stabilizes when reaching a cycle. In particular, substituting $\bot$ for $p$ in $A^b(p)$ produces the least fixpoint, while substituting $\top$ produces the greatest fixpoints. This is why Mardaev \cite{Mardaev1993:al,Mardaev1994:al,Mardaev1997:bsl,Mardaev1998:al,Mardaev2007:jancl} quotes Ruitenburg when investigating the issue of fixpoint definability in non-classical logics.  Periodic sequences, however, are by no means the only way of ensuring that a logic has definable fixpoints: there are examples of systems having the latter property without the former. In fact, a combination of Pitts' \cite{Pitts92:jsl} uniform interpolation  with what modal logicians would describe as a (definable) \emph{master modality}\footnote{This is only needed if the logic in question contains additional ``modal'' connectives or lacks some structural rules. For intutionistic propositional logic itself, the requirement of having a ``master modality'' (global deduction theorem, equationally definable principal congruences...) is trivially satisfied, just like for standard relevance logics. On the other hand, 
 this criterion is not generally met by substructural logics such as those covered by Theorem \ref{th:square}. They generally fail to 
  satisfy axioms ensuring EDPC \cite[Theorem 3.55]{GalatosJKO07}. Same problems arise with non-transitive modalities, even in the classical unimodal setting.} plus some trivial additional restrictions would be sufficient.  Ghilardi et al. \cite{Ghilardi2016, GhilardiGS20} discuss further the issue of computing fixpoints and fixpoint definability, and compare the two approaches. It is worth mentioning here that Ruitenburg in the final part of his paper suggests a potential connection with uniform interpolation, despite preceding  Pitts' \cite{Pitts92:jsl} by several years. 
 
Finally, as suggested by a referee, a clarification might be in order. The reader might be aware of results on definability of fixpoints in modal logics (classical or intuitionistic ones) containing some form of the L\"ob axiom. Nevertheless, as pointed out in Corollaries \ref{cor:modclass} and \ref{cor:intclas}, such logics generally fail the plps property. Such definability results concern \emph{guarded} or \emph{modalized} fixpoints, i.e., those where the fixpoint variable occurs only in the scope of a modal operator. Van Benthem \cite{Benthem06:sl} and Visser \cite{Visser05:lncs} illustrate how to use definability of such fixpoints to derive definability of ordinary, monotone fixpoints: Precisely here one can use  the periodic sequence property for the underlying modality-free (\emph{extensional}) reduct of the logic in question.

The purpose of this note is twofold. In \Cref{sec:context}, I discuss the status of periodic sequences in other non-classical logics,  illustrating just how special the situation of \ipc\ is. Starting from \Cref{sec:formalization}, I present a mechanization of Ruitenburg's result in the Coq proof assistant.\footnote{The content of both parts is based on the work done in years 2015--2017, which for various reasons remained unpublished and presented only in the form of a talk at TACL 2017.}

\subsection{Related work} \label{sec:related}
In 2015--2016, when the mechanization described herein 
 was produced, Ruitenburg's original and poorly understood syntactic proof\footnote{It is worth mentioning here that Wim Ruitenburg himself (p.c.) claims that his original proof was utilizing Kripke semantics. Difficulties in explaining it to his colleagues, in particular Albert Visser, and Visser's additional insights finally recorded as Lemma 1.7 in Ruitenburg's paper, convinced him to cast the argument into a purely syntactic setting, which at the time proved clear enough to both Ruitenburg and Visser.}  was the only available one. This in fact motivated the present author to take on the challenge of mechanizing the proof, despite a relatively limited experience with proof assistants at the time. In the meantime, Ghilardi and Santocanale \cite{GhilardiS18,GhilardiS20} provided a semantic proof using the apparatus developed in the Ghilardi and Zawadowski monograph \cite{ghil:shea02}, involving games for bounded bisimulations and (pre)sheaves over the category of finite rooted posets with bounded morphisms. Nevertheless, as Ghilardi and Santocanale admit, their semantic proof does not provide tight bounds and computational information provided by Ruitenburg's proof, and extracted by the Coq mechanization described in this paper. Furthermore, the hope expressed in their final remark
\begin{quote}
While we can expect that periodicity
phenomena of substitutions do not arise for the basic modal logic \lname{K}, they
surely do for locally tabular [i.e., locally finite] modal logics. Considering also the numerous
results on definability of fixpoints \dots these phenomena are
likely to appear in other subsystems of modal logics. As far as we know,
investigation of periodicity phenomena in modal logics is a research direction which has not yet been explored and where the bounded bisimulation
methods might prove their strength once more.
\end{quote}
in the light of \Cref{sec:context} herein requires qualification: outside of locally finite modal logics, there seems to be little hope and scope for periodicity. Open Problem \ref{probl:pll} below isolates one potential intuitionistic modal logic for which a generalization of Ruitenburg result might be possible. In fact, the mechanization described here can be used in investigating the problem or (should the answer turns out to be positive) verifying a potential proof; cf. \Cref{rem:extensions} in \Cref{sec:setup} and \Cref{rem:pll} in \Cref{sec:aux}.

Of all Coq formalizations of non-trivial results concerning various propositional calculi, the recent work of F{\'{e}}r{\'{e}}e and van Gool \cite{FereeG23} is probably closest to our interests here. It deals with Pitt's syntactic proof of uniform interpolation for \ipc, which as discussed above provides another route towards fixpoint definability, and it also allows extraction of executable code, actually computing propositional quantifiers\footnote{To make the relationship even stronger, the apparatus developed in the Ghilardi and Zawadowski monograph \cite{ghil:shea02} provides a model-theoretic proof of both Pitts' result and Ruitenburg's result. In fact, the starting point for that monograph was their earlier article \cite{ghil:shea95}, explicitly motivated as ``language-free'' or categorical analysis of uniform interpolation. Visser  \cite{viss:laye96} follows a similar approach based on bounded bisimulations, cast in somewhat less categorical terms.}. As  Pitts' proof was cast in the setting of the terminating sequent calculus \lname{G4ip} \cite{hude:boun89,dyck:cont92} (Ruitenburg, by contrast, works with a slightly idiosyncratic and purely Hilbert-style setting, as discussed in \Cref{sec:formalization}),  F{\'{e}}r{\'{e}}e and van Gool \cite{FereeG23} mechanizes some metatheory of that calculus, in particular admissibility of a restricted form of cut and other structural rules. The present mechanization simply assumes decidability of \ipc\ and does not attempt to provide either a syntactic proof via cut elimination or a Kripke-based semantic one, although such developments are available elsewhere.  

Finally, there is an entire recent body of work mechanizing in Coq \lname{G4ip}-style calculi for several propositional logics (over classical and intuitionistic base) developed by Shillito and coauthors \cite{ShillitoGGI23,ShillitoG22,GoreRS21}. It would seem of interest to turn the present formalization into a part of a larger library, integrating the developments described above and possibly other scattered contributions, such as the Coq development supporting the discussion of modal negative translations in Litak et al. \cite{litakpr17}. 




\section{Periodic sequences in nonclassical logics} \label{sec:nonclassical} \label{sec:context}


\ifdef{\longv}{
In order to understand properly the special status of Ruitenburg's result, let us compare the situation in \ipc\ with that in other non-classical logics, e.g., substructural or modal ones. It turns out that in almost all standard cases,  plps (and even more so, ulps) implies local finiteness; \ipc\ seems rather exotic in having the first property without the second one. 




\subsection{Modal logics over \cpc}

For modal logics over the boolean propositional base, the reader can refer to, e.g., Chagrov and Zakharyaschev \cite{ChagrovZ97:ml} for notation, syntax and semantics; one difference is that I am using here a superscript \cln{\cdot} to make the \cpc\ propositional base clear.  For transitive modal logics, having periodic sequences is indistinguishable from local finiteness, i.e., the converse of Lemma \ref{lem:llf} holds:}
{As it turns out, however, finding other natural examples of logic enjoying plps without local finiteness  is a very challenging task. First let us consider intuitionistic or classical normal modal logics (with $\Box$ only), with  superscript \cln{\cdot} denoting the \cpc\ propositional base:}

\begin{theorem} \label{th:transfail}
A normal extension of \cln{K4} has plps iff it is locally finite.
\end{theorem}

\ifdef{\longv}{
\begin{proof}
 It is known \cite[Theorem 12.21]{ChagrovZ97:ml} that a normal modal logic extending \lname{K4} is locally finite iff it is of \emph{infinite depth}, i.e., admits Kripke frames of arbitrary finite depths. Consider $A_\lname{K4}(p) \deq q \vee \Box (q \to \Box p)$. A straightforward modification of the argument proving the above equivalence \cite[Theorem 12.21]{ChagrovZ97:ml} shows the failure of plps (in the proof, the valuation for $q$  should be defined in the same way as the valuation for $p$): the sequence $\{A_\lname{K4}^n(p)\}_{n \in \omega}$ never stabilizes.
\end{proof}}{}

\ifdef{\longv}{
\begin{corollary} \label{cor:modclass}
All extensions of \cln{K} contained in either \cln{S4Grz.3}   (such as  \cln{K4},  \cln{S4},  \cln{T}) or \cln{GL.3} \bro in particular \cln{GL}\brc fail to have locally periodic sequences.
\end{corollary} 

For subsystems of \cln{GL.3}, this can be proved via a simpler alternative technique that remains useful when the propositional base is weakened to \ipc; see Theorem \ref{th:sfail}.

Moreover, even without transitivity it does not appear easy to find examples of logics with plps which are not locally finite. Shapirovsky \cite{Shapirovsky18:aiml} has provided an example of a normal modal logic which has finitely many formulas in one variable, but fails local finiteness. Unfortunately, the technique used in the proof of Theorem \ref{th:transfail}  can be also applied to his example with a minor modification: namely, use $$A_\lname{Sha}(p) \deq q \vee \Box (q \to \Box (p \vee r)),$$ where $r$ is to be evaluated as $\{\omega\}$ in Shapirovsky's frame.

\subsection{Intuitionistic modal logics}

The reader is referred to the extensive literature \cite{Simpson94:phd,WolterZ97:al,WolterZ98:lw,Litak14:trends} for basic information about intuitionistic modal logics. Just for clarity,  we are only discussing here intuitionistic modal logics with a single modality $\Box$ (no $\Diamond$), which is reflected in the notation. 
Theorem \ref{th:transfail} immediately extends to intuitionistic modal logics being counterparts of standard extensions of \cln{K} (see Simpson's \cite{Simpson94:phd} Requirement 3):

\begin{corollary} \label{cor:intclas}
All extensions of \iln{K} contained in either \cln{S4Grz.3}   \bro such as \iln{K4},  \iln{T},\iln{S4},   \iln{S4Grz.3} or  \iln{S4Grz.3}\brc or \cln{GL.3} \bro in particular \iln{GL} or \iln{GL.3}\brc fail to have locally periodic sequences.
\end{corollary} 
}
{\begin{corollary}  \label{cor:intclas}
All extensions of \iln{K} contained in either \cln{S4Grz.3} \bro including, for example, \cln{K}, \cln{K4},  \cln{S4},  \cln{T}, \iln{K4},  \iln{T}, \iln{S4},   \iln{S4Grz.3} or  \iln{S4Grz.3}\brc or \cln{GL.3} \bro including, for example, \cln{GL}, \iln{GL} or \iln{GL.3}\brc fail to have locally periodic sequences.\footnote{The reader is referred to the extensive literature \cite{BellinPR01:m4m,Simpson94:phd,Wolter97:sl,WolterZ97:al,WolterZ98:lw,Litak14:trends} for basic information about intuitionistic modal logics, including axiomatizations of systems mentioned in this theorem.}
\end{corollary} }

Some intuitionistic modal logics of  computational interest have ``degenerate'' classical counterparts \ifdef{\longv}{(see \cite{Litak14:trends} for a discussion)}{} and hence Corollary \ref{cor:intclas} cannot be used to disprove that they have periodic sequences. This includes $\iln{S} \deq \iln{K} \oplus A \to \Box A$, i.e.,  the Curry-Howard logic of \emph{applicative functors}, also known as \emph{idioms} \cite{McBrideP08:jfp}. Its 
 classical counterpart \cln{S}\ and all its two consistent proper 
 extensions are 
 finite logics enjoying ulps. 
 \ifdef{\longv}{
  In fact, \cln{S}  has exactly two proper consistent extensions, one 
  denoted as \lname{Triv} 
  and the other 
  denoted as \lname{Ver} \cite{ChagrovZ97:ml}.}{}  
 In contrast, not only does  \iln{S}\ have uncountably many propositional extensions, but the failure of plps remains a common phenomenon among  them\ifdef{\longv}{. To show this, one can use a proof technique applicable to (subsystems of) logics with L\"ob-style axioms, either intuitionistic or classical ones}{}: 


\begin{theorem} \label{th:sfail}
No sublogic of  $\iln{KM.3}$, also denoted as $\lname{KM_{lin}}$ \textup{\cite{CloustonG15:fossacs}} has parametrically locally periodic sequences; this in particular applies to  $\iln{SL.3} \deq \iln{S} \oplus \iln{GL.3}$, $\iln{SL} \deq \iln{S} \oplus \iln{GL}$ or \iln{S}.
\end{theorem}

\ifdef{\longv}{
\begin{proof}
 The logic $\iln{KM.3}$ (or $\lname{KM_{lin}}$) is the logic of the Kripke frame where the modal and the intuitionistic order are, respectively, irreflexive and reflexive variant of the reverse order on natural numbers. Consider $A_\lname{KM}(p) \deq \Box p$ and the valuation sending $p$ to $\emptyset$; the denotation of $A_\lname{KM}^n(p)$ is the set of natural numbers smaller or equal to $n$. Hence, the sequence $\{A_\lname{KM}^n(p)\}_{n \in \omega}$ never stabilizes and  plps fails  in every logic sound in this frame.
 \end{proof}}{}

To contrast this with Theorem \ref{th:transfail}, note that $\iln{KM.3}$, the propositional fragment of the logic of the Mitchell-B\'enabou logic of the \emph{topos of trees} \cite{BirkedalMSS12:lmcs,CloustonG15:fossacs,Litak14:trends}, is \ifdef{\longv}{\emph{prefinite} or \emph{pretabular}: all its extensions are finite, each determined by a finite chain.
  Interestingly, neither the proof Theorem \ref{th:transfail} nor the proof of Theorem \ref{th:sfail} apply to the Propositional Lax Logic \iln{PLL} \cite{FairtloughM97:ic}, i.e., the Curry-Howard counterpart to (the type system of) Moggi's monadic metalanguage \cite{BentonBP98:jfp}.  

\begin{probl} \label{probl:pll}
Does \iln{PLL} have locally periodic sequences?
\end{probl}}
{ prefinite (\emph{pretabular}). Turning to substructural logics:}

\subsection{Substructural logics}

\newcommand{\fus}{\cdot}
\newcommand{\luka}[1]{\ensuremath{\text{\sf \L}_{#1}}}

Arguments analogous to those above establish that in the realm of substructural logics \cite{GalatosJKO07}, plps as a rule coincides with local finiteness. Consider $A_\otimes(p) \deq p \fus p$, where $\fus$ is the substructural \emph{fusion} connective \cite[\S 2.1.2]{GalatosJKO07}, also known by linear logicians as \emph{tensor} or \emph{multiplicative conjunction} $\otimes$, and in the realm of the Logic of Bunched Implications \lname{BI} and separation logic as \emph{spatial}, \emph{separating} or \emph{independent} conjunction $*$ \cite{Reynolds02:lics}. 

\begin{theorem} \label{th:square}
 The product logic \lname{\Pi},  the infinite valued \L ukasiewicz logic \luka{\infty} or the logic of the \emph{heap model} of \lname{BBI} \bro boolean logic of bunched implications \textup{\cite{BrotherstonK14:jacm,OHearnP99:jsl,PymOHY04:tcs,Reynolds02:lics}}\brc fail to have plps.  Consequently, the  property fails in all their sublogics, including \lname{(In-)FL_{(ew)}}, multiplicative-additive fragment of linear logic \lname{MALL} \bro and its intuitionistic fragment \lname{IMALL}\brc and fuzzy logics such as \lname{BL} or \lname{MTL}.\footnote{See Galatos et al. \cite{GalatosJKO07} for substructural systems mentioned in the statement of this theorem.}
\end{theorem}

\begin{proof}
This is shown by evaluating the sequence $\{A_\otimes^n(p)\}_{n \in \omega}$ defined above in the heap model or  the $[0,1]$-interval with corresponding $\ell$-norms. 
\end{proof}

Thus, in order to find a natural substructural logic $L$ enjoying the plps without local finiteness, one should look at those where the sequence $\{A_\otimes^n(p)\}_{n \in \omega}$ stabilizes modulo $\eql{L}$. It can be  naturally achieved by stipulating that fusion is  idempotent (both square-increasing and square-decreasing). This, however, is a very restrictive condition. When $L$ satisfies the weakening rule, it collapses fusion to ordinary additive conjunction $\wedge$ and substructural implication to Heyting implication. Idempotent systems where $\fus$ does not entirely collapse to $\wedge$ are sometimes considered by relevance logicians, with perhaps the most famous example being  \lname{RM} \bro``\lname{R} with Mingle''\brc. However, this system has been long known to be locally finite anyway \cite{Dunn70}. See more recent references \cite{Raftery07,GilFJM20} on limits of local finiteness results for idempotent structural logics. 

\begin{probl} \label{probl:relevance}
Are there natural non-Heyting examples of (idempotent? square-increasing?) substructural logics with the plps property failing local finiteness? 
\end{probl}

%

\section{Coq Formalization: The Basics} \label{sec:formalization}

The formalization is available as a git repository 
\begin{center}
\url{https://git8.cs.fau.de/software/ruitenburg1984}. 
\end{center}
It consists of less than 4000 lines of Coq code, split into 6 files.  The code allows working program extraction to OCaml or Haskell; it can be also used directly for computation using Coq's core functional programming language (Gallina). More on that can be found in \Cref{sec:bounds_lists}.

Naturally, the formalization involves a \emph{deep} rather than \emph{shallow} embedding of \ipc. The syntax of \ipc\ is formalized from first principles. Also, all semantic discussion (i.e., any mention of Kripke models) from the original paper is omitted. The semantic counterexamples given by Ruitenburg are unproblematic and easy to understand. The important part is purely syntactic.




\subsection{Setup and basic lemmas} \label{sec:setup}

The language of \ipc\  is defined as usual:



\begin{coqdoccode}
\coqdocemptyline
\coqdocnoindent
\coqdockw{Inductive}  \coqdef{HilbertIPCsetup.form}{form}{\coqdocinductive{form}}   :=\coqdoceol
\coqdocindent{1.00em}
\ensuremath{|}  \coqdef{HilbertIPCsetup.var}{var}{\coqdocconstructor{var}} :  \coqexternalref{nat}{http://coq.inria.fr/distrib/8.4pl6/stdlib/Coq.Init.Datatypes}{\coqdocinductive{nat}} \ensuremath{\rightarrow} \coqref{HilbertIPCsetup.form}{\coqdocinductive{form}}\coqdoceol
\coqdocindent{1.00em}
\ensuremath{|}  \coqdef{HilbertIPCsetup.imp}{imp}{\coqdocconstructor{imp}} :  \coqref{HilbertIPCsetup.form}{\coqdocinductive{form}} \ensuremath{\rightarrow} \coqref{HilbertIPCsetup.form}{\coqdocinductive{form}} \ensuremath{\rightarrow} \coqref{HilbertIPCsetup.form}{\coqdocinductive{form}}\coqdoceol
\coqdocindent{1.00em}
\ensuremath{|}  \coqdef{HilbertIPCsetup.and}{and}{\coqdocconstructor{and}} :  \coqref{HilbertIPCsetup.form}{\coqdocinductive{form}} \ensuremath{\rightarrow} \coqref{HilbertIPCsetup.form}{\coqdocinductive{form}} \ensuremath{\rightarrow} \coqref{HilbertIPCsetup.form}{\coqdocinductive{form}}\coqdoceol
\coqdocindent{1.00em}
\ensuremath{|}  \coqdef{HilbertIPCsetup.or}{or}{\coqdocconstructor{or}}  :  \coqref{HilbertIPCsetup.form}{\coqdocinductive{form}} \ensuremath{\rightarrow} \coqref{HilbertIPCsetup.form}{\coqdocinductive{form}} \ensuremath{\rightarrow} \coqref{HilbertIPCsetup.form}{\coqdocinductive{form}}\coqdoceol
\coqdocindent{1.00em}
\ensuremath{|}  \coqdef{HilbertIPCsetup.tt}{tt}{\coqdocconstructor{tt}} :  \coqref{HilbertIPCsetup.form}{\coqdocinductive{form}}\coqdoceol
\coqdocindent{1.00em}
\ensuremath{|}  \coqdef{HilbertIPCsetup.ff}{ff}{\coqdocconstructor{ff}} :  \coqref{HilbertIPCsetup.form}{\coqdocinductive{form}}.\coqdoceol
\coqdocemptyline
\coqdocnoindent
\coqdockw{Notation} \coqdef{HilbertIPCsetup.::x 'x26' x}{"}{"}A '\&' B " := (\coqref{HilbertIPCsetup.and}{\coqdocconstructor{and}} \coqdocvar{A} \coqdocvar{B}) 
(\coqdoctac{at} \coqdockw{level} 40, \coqdoctac{left} \coqdockw{associativity}).\coqdoceol
\coqdocnoindent
\coqdockw{Notation} \coqdef{HilbertIPCsetup.::x 'x5Cv/' x}{"}{"}A '\symbol{92}v/' B " := (\coqref{HilbertIPCsetup.or}{\coqdocconstructor{or}} \coqdocvar{A} \coqdocvar{B}) 
(\coqdoctac{at} \coqdockw{level} 45, \coqdoctac{left} \coqdockw{associativity}).\coqdoceol
\coqdocnoindent
\coqdockw{Notation} \coqdef{HilbertIPCsetup.::x '->>' x}{"}{"}A '->>' B" := (\coqref{HilbertIPCsetup.imp}{\coqdocconstructor{imp}} \coqdocvar{A} \coqdocvar{B})
 (\coqdoctac{at} \coqdockw{level} 49, \coqdoctac{right} \coqdockw{associativity}).\coqdoceol
\coqdocemptyline
\coqdocemptyline
\coqdocnoindent
\coqdockw{Definition} \coqdef{HilbertIPCsetup.p}{p}{\coqdocdefinition{p}} := \coqref{HilbertIPCsetup.var}{\coqdocconstructor{var}} 0.\coqdoceol
\coqdocnoindent
\coqdockw{Definition} \coqdef{HilbertIPCsetup.q}{q}{\coqdocdefinition{q}} := \coqref{HilbertIPCsetup.var}{\coqdocconstructor{var}} 1.\coqdoceol
\coqdocnoindent
\coqdockw{Definition} \coqdef{HilbertIPCsetup.r}{r}{\coqdocdefinition{r}} := \coqref{HilbertIPCsetup.var}{\coqdocconstructor{var}} 2.\coqdoceol
\coqdocemptyline
\end{coqdoccode}

\noindent
Variable $p$ will be used as a distinguished variable of the formula: the input or the argument of a polynomial. It is convenient to make it the variable with index 0 (consider the use of \coqdockw{destruct} in proofs where the distinguished variable should get a special treatment). We also explicitly add equivalence: 

\begin{coqdoccode}
\coqdocemptyline
\coqdockw{Notation} \coqdef{HilbertIPCsetup.::x '<<->>' x}{"}{"}A '<<->>' B" := (\coqref{HilbertIPCsetup.::x 'x26' x}{\coqdocnotation{(}}\coqdocvar{A} \coqref{HilbertIPCsetup.::x '->>' x}{\coqdocnotation{->>}} \coqdocvar{B}\coqref{HilbertIPCsetup.::x 'x26' x}{\coqdocnotation{)}} \coqref{HilbertIPCsetup.::x 'x26' x}{\coqdocnotation{\&}} \coqref{HilbertIPCsetup.::x 'x26' x}{\coqdocnotation{(}}\coqdocvar{B} \coqref{HilbertIPCsetup.::x '->>' x}{\coqdocnotation{->>}} \coqdocvar{A}\coqref{HilbertIPCsetup.::x 'x26' x}{\coqdocnotation{)}}) (\coqdoctac{at} \coqdockw{level} 58).\coqdoceol
\coqdocemptyline
\end{coqdoccode}

\noindent
Ruitenburg's paper uses a formulation of \ipc\ in terms of syntactic consequence (turnstile) relation between (finite) sets of formulas and formulas themselves. This approach is natural from the point of view of abstract algebraic logic  \cite{BlokP89:ams,FontJP03a:sl}. It would be natural to replace this turnstile-Hilbert-style axiomatization by a Gentzen-style formalism, either sequent calculus or natural deduction. We will return to this point in \Cref{sec:turnstile}. 
 The chosen formalization of \ipc, however, is convenient for our purposes and stays as close to the development in the original article \cite{Ruitenburg84:jsl} as possible (Ruitenburg, in fact, did not write the exact axiomatization he was using or give an explicit reference for it, but it is easy to reconstruct). 

Rather unsurprisingly, in the Coq version of the axiomatization I replaced finite sets of formulas with finite lists. The standard Hilbert-style presentation of \ipc\ can be found in numerous references.  In my setup, it looks as follows:

\begin{coqdoccode}
\coqdocemptyline
\coqdocnoindent
\coqdockw{Reserved Notation} "G '|--' A" (\coqdoctac{at} \coqdockw{level} 63).\coqdoceol
\coqdocemptyline
\coqdocnoindent
\coqdockw{Notation} \coqdef{HilbertIPCsetup.context}{context}{\coqdocabbreviation{context}} := (\coqexternalref{list}{http://coq.inria.fr/distrib/8.4pl6/stdlib/Coq.Init.Datatypes}{\coqdocinductive{list}} \coqref{HilbertIPCsetup.form}{\coqdocinductive{form}}).\coqdoceol
\coqdocemptyline
\coqdocnoindent
\coqdockw{Inductive} \coqdef{HilbertIPCsetup.hil}{hil}{\coqdocinductive{hil}} : \coqref{HilbertIPCsetup.context}{\coqdocabbreviation{context}} \ensuremath{\rightarrow} \coqref{HilbertIPCsetup.form}{\coqdocinductive{form}} \ensuremath{\rightarrow} \coqdockw{Prop} :=\coqdoceol
\coqdocindent{0.50em}
\ensuremath{|} \coqdef{HilbertIPCsetup.hilst}{hilst}{\coqdocconstructor{hilst}} : \coqdockw{\ensuremath{\forall}} \coqdocvar{G} \coqdocvar{A},  \coqexternalref{In}{http://coq.inria.fr/distrib/8.4pl6/stdlib/Coq.Lists.List}{\coqdocdefinition{In}} \coqdocvariable{A} \coqdocvariable{G} \ensuremath{\rightarrow} \coqdocvariable{G} \coqref{HilbertIPCsetup.::x '|--' x}{\coqdocnotation{|--}} \coqdocvariable{A} \coqdoceol
\coqdocindent{0.50em}
\ensuremath{|} \coqdef{HilbertIPCsetup.hilK}{hilK}{\coqdocconstructor{hilK}}  : \coqdockw{\ensuremath{\forall}} \coqdocvar{G} \coqdocvar{A} \coqdocvar{B}, \coqdocvariable{G} \coqref{HilbertIPCsetup.::x '|--' x}{\coqdocnotation{|--}} \coqdocvariable{A} \coqref{HilbertIPCsetup.::x '->>' x}{\coqdocnotation{->>}} \coqdocvariable{B} \coqref{HilbertIPCsetup.::x '->>' x}{\coqdocnotation{->>}} \coqdocvariable{A}\coqdoceol
\coqdocindent{0.50em}
\ensuremath{|} \coqdef{HilbertIPCsetup.hilS}{hilS}{\coqdocconstructor{hilS}}  : \coqdockw{\ensuremath{\forall}} \coqdocvar{G} \coqdocvar{A} \coqdocvar{B} \coqdocvar{C}, 
\coqdocvariable{G} \coqref{HilbertIPCsetup.::x '|--' x}{\coqdocnotation{|--}} \coqref{HilbertIPCsetup.::x '->>' x}{\coqdocnotation{(}}\coqdocvariable{A} \coqref{HilbertIPCsetup.::x '->>' x}{\coqdocnotation{->>}} \coqdocvariable{B} \coqref{HilbertIPCsetup.::x '->>' x}{\coqdocnotation{->>}} \coqdocvariable{C}\coqref{HilbertIPCsetup.::x '->>' x}{\coqdocnotation{)}} \coqref{HilbertIPCsetup.::x '->>' x}{\coqdocnotation{->>}} \coqref{HilbertIPCsetup.::x '->>' x}{\coqdocnotation{(}}\coqdocvariable{A} \coqref{HilbertIPCsetup.::x '->>' x}{\coqdocnotation{->>}} \coqdocvariable{B}\coqref{HilbertIPCsetup.::x '->>' x}{\coqdocnotation{)}} \coqref{HilbertIPCsetup.::x '->>' x}{\coqdocnotation{->>}} \coqdocvariable{A} \coqref{HilbertIPCsetup.::x '->>' x}{\coqdocnotation{->>}} \coqdocvariable{C}\coqdoceol
\coqdocindent{0.50em}
\ensuremath{|} \coqdef{HilbertIPCsetup.hilMP}{hilMP}{\coqdocconstructor{hilMP}} : \coqdockw{\ensuremath{\forall}} \coqdocvar{G} \coqdocvar{A} \coqdocvar{B}, \coqdocvariable{G} \coqref{HilbertIPCsetup.::x '|--' x}{\coqdocnotation{|--}} \coqref{HilbertIPCsetup.::x '|--' x}{\coqdocnotation{(}}\coqdocvariable{A} \coqref{HilbertIPCsetup.::x '->>' x}{\coqdocnotation{->>}} \coqdocvariable{B}\coqref{HilbertIPCsetup.::x '|--' x}{\coqdocnotation{)}} \ensuremath{\rightarrow} (\coqdocvariable{G} \coqref{HilbertIPCsetup.::x '|--' x}{\coqdocnotation{|--}} \coqdocvariable{A}) \ensuremath{\rightarrow} (\coqdocvariable{G} \coqref{HilbertIPCsetup.::x '|--' x}{\coqdocnotation{|--}} \coqdocvariable{B})\coqdoceol
\coqdocindent{0.50em}
\ensuremath{|} \coqdef{HilbertIPCsetup.hilC1}{hilC1}{\coqdocconstructor{hilC1}} : \coqdockw{\ensuremath{\forall}} \coqdocvar{G} \coqdocvar{A} \coqdocvar{B}, \coqdocvariable{G} \coqref{HilbertIPCsetup.::x '|--' x}{\coqdocnotation{|--}} \coqref{HilbertIPCsetup.::x '|--' x}{\coqdocnotation{(}}\coqdocvariable{A} \coqref{HilbertIPCsetup.::x '->>' x}{\coqdocnotation{->>}} \coqdocvariable{B} \coqref{HilbertIPCsetup.::x '->>' x}{\coqdocnotation{->>}} \coqdocvariable{A} \coqref{HilbertIPCsetup.::x 'x26' x}{\coqdocnotation{\&}} \coqdocvariable{B}\coqref{HilbertIPCsetup.::x '|--' x}{\coqdocnotation{)}}\coqdoceol
\coqdocindent{0.50em}
\ensuremath{|} \coqdef{HilbertIPCsetup.hilC2}{hilC2}{\coqdocconstructor{hilC2}} : \coqdockw{\ensuremath{\forall}} \coqdocvar{G} \coqdocvar{A} \coqdocvar{B}, \coqdocvariable{G} \coqref{HilbertIPCsetup.::x '|--' x}{\coqdocnotation{|--}} \coqref{HilbertIPCsetup.::x '|--' x}{\coqdocnotation{(}}\coqdocvariable{A} \coqref{HilbertIPCsetup.::x 'x26' x}{\coqdocnotation{\&}} \coqdocvariable{B} \coqref{HilbertIPCsetup.::x '->>' x}{\coqdocnotation{->>}} \coqdocvariable{A}\coqref{HilbertIPCsetup.::x '|--' x}{\coqdocnotation{)}}\coqdoceol
\coqdocindent{0.50em}
\ensuremath{|} \coqdef{HilbertIPCsetup.hilC3}{hilC3}{\coqdocconstructor{hilC3}} : \coqdockw{\ensuremath{\forall}} \coqdocvar{G} \coqdocvar{A} \coqdocvar{B}, \coqdocvariable{G} \coqref{HilbertIPCsetup.::x '|--' x}{\coqdocnotation{|--}} \coqref{HilbertIPCsetup.::x '|--' x}{\coqdocnotation{(}}\coqdocvariable{A} \coqref{HilbertIPCsetup.::x 'x26' x}{\coqdocnotation{\&}} \coqdocvariable{B} \coqref{HilbertIPCsetup.::x '->>' x}{\coqdocnotation{->>}} \coqdocvariable{B}\coqref{HilbertIPCsetup.::x '|--' x}{\coqdocnotation{)}}\coqdoceol
\coqdocindent{0.50em}
\ensuremath{|} \coqdef{HilbertIPCsetup.hilA1}{hilA1}{\coqdocconstructor{hilA1}} : \coqdockw{\ensuremath{\forall}} \coqdocvar{G} \coqdocvar{A} \coqdocvar{B}, \coqdocvariable{G} \coqref{HilbertIPCsetup.::x '|--' x}{\coqdocnotation{|--}} \coqref{HilbertIPCsetup.::x '|--' x}{\coqdocnotation{(}}\coqdocvariable{A} \coqref{HilbertIPCsetup.::x '->>' x}{\coqdocnotation{->>}} \coqdocvariable{A} \coqref{HilbertIPCsetup.::x 'x5Cv/' x}{\coqdocnotation{\symbol{92}}}\coqref{HilbertIPCsetup.::x 'x5Cv/' x}{\coqdocnotation{v}}\coqref{HilbertIPCsetup.::x 'x5Cv/' x}{\coqdocnotation{/}} \coqdocvariable{B}\coqref{HilbertIPCsetup.::x '|--' x}{\coqdocnotation{)}}\coqdoceol
\coqdocindent{0.50em}
\ensuremath{|} \coqdef{HilbertIPCsetup.hilA2}{hilA2}{\coqdocconstructor{hilA2}} : \coqdockw{\ensuremath{\forall}} \coqdocvar{G} \coqdocvar{A} \coqdocvar{B}, \coqdocvariable{G} \coqref{HilbertIPCsetup.::x '|--' x}{\coqdocnotation{|--}} \coqref{HilbertIPCsetup.::x '|--' x}{\coqdocnotation{(}}\coqdocvariable{B} \coqref{HilbertIPCsetup.::x '->>' x}{\coqdocnotation{->>}} \coqdocvariable{A} \coqref{HilbertIPCsetup.::x 'x5Cv/' x}{\coqdocnotation{\symbol{92}}}\coqref{HilbertIPCsetup.::x 'x5Cv/' x}{\coqdocnotation{v}}\coqref{HilbertIPCsetup.::x 'x5Cv/' x}{\coqdocnotation{/}} \coqdocvariable{B}\coqref{HilbertIPCsetup.::x '|--' x}{\coqdocnotation{)}}\coqdoceol
\coqdocindent{0.50em}
\ensuremath{|} \coqdef{HilbertIPCsetup.hilA3}{hilA3}{\coqdocconstructor{hilA3}} : \coqdockw{\ensuremath{\forall}} \coqdocvar{G} \coqdocvar{A} \coqdocvar{B} \coqdocvar{C}, 
\coqdocvariable{G} \coqref{HilbertIPCsetup.::x '|--' x}{\coqdocnotation{|--}} \coqref{HilbertIPCsetup.::x '->>' x}{\coqdocnotation{(}}\coqdocvariable{A} \coqref{HilbertIPCsetup.::x '->>' x}{\coqdocnotation{->>}} \coqdocvariable{C}\coqref{HilbertIPCsetup.::x '->>' x}{\coqdocnotation{)}} \coqref{HilbertIPCsetup.::x '->>' x}{\coqdocnotation{->>}} \coqref{HilbertIPCsetup.::x '->>' x}{\coqdocnotation{(}}\coqdocvariable{B} \coqref{HilbertIPCsetup.::x '->>' x}{\coqdocnotation{->>}} \coqdocvariable{C}\coqref{HilbertIPCsetup.::x '->>' x}{\coqdocnotation{)}} \coqref{HilbertIPCsetup.::x '->>' x}{\coqdocnotation{->>}} \coqref{HilbertIPCsetup.::x '->>' x}{\coqdocnotation{(}}\coqdocvariable{A} \coqref{HilbertIPCsetup.::x 'x5Cv/' x}{\coqdocnotation{\symbol{92}}}\coqref{HilbertIPCsetup.::x 'x5Cv/' x}{\coqdocnotation{v}}\coqref{HilbertIPCsetup.::x 'x5Cv/' x}{\coqdocnotation{/}} \coqdocvariable{B} \coqref{HilbertIPCsetup.::x '->>' x}{\coqdocnotation{->>}} \coqdocvariable{C}\coqref{HilbertIPCsetup.::x '->>' x}{\coqdocnotation{)}}\coqdoceol
\coqdocindent{0.50em}
\ensuremath{|} \coqdef{HilbertIPCsetup.hiltt}{hiltt}{\coqdocconstructor{hiltt}} : \coqdockw{\ensuremath{\forall}} \coqdocvar{G}, \coqdocvariable{G} \coqref{HilbertIPCsetup.::x '|--' x}{\coqdocnotation{|--}} \coqref{HilbertIPCsetup.tt}{\coqdocconstructor{tt}}\coqdoceol
\coqdocindent{0.50em}
\ensuremath{|} \coqdef{HilbertIPCsetup.hilff}{hilff}{\coqdocconstructor{hilff}} : \coqdockw{\ensuremath{\forall}} \coqdocvar{G} \coqdocvar{A}, \coqdocvariable{G} \coqref{HilbertIPCsetup.::x '|--' x}{\coqdocnotation{|--}} \coqref{HilbertIPCsetup.ff}{\coqdocconstructor{ff}} \coqref{HilbertIPCsetup.::x '->>' x}{\coqdocnotation{->>}} \coqdocvariable{A}\coqdoceol
\coqdocindent{9.50em}
\coqdockw{where} \coqdef{HilbertIPCsetup.::x '|--' x}{"}{"}G '|--' A" := (\coqref{HilbertIPCsetup.hil}{\coqdocinductive{hil}} \coqdocvar{G} \coqdocvar{A}).\coqdoceol
\end{coqdoccode}

A sequence of lemmas follows in \coqdockw{HilbertIPCsetup.v}, establishing basic properties of the turnstile relation. There are also some easy tactics for use in later proofs. They should be all rather self-explanatory. Again, in Ruitenburg's paper trivial lemmas of this kind are used tacitly or nearly tacitly. 
 Several, though by no means all of the basic lemmas in this part  were added to the \cdkw{Hint} database and so were, e.g., constructors of \coqdoclemma{hil}. This has in some cases slowed down working of some tactics, in particular \cdkw{eauto}, but I believe overall the database has not been unduly swollen and \cdkw{eauto}, \cdkw{auto} and their cousins remain useful.

One also needs a standard notion of substitution; as the focus is entirely on a propositional language with  no notions of---and no problems of---binding, $\alpha$-conversion etc., I decided to use a straightforward, own formalization, with a minimal number of tailored tactics to make the development smoother. Base substitutions are simply functions from variables to formulas; we can thus reduce them to functions from natural numbers to formulas. They are extended inductively to arbitrary formulas and then to arbitrary contexts:

\begin{coqdoccode}
\coqdocemptyline
\coqdocnoindent
\coqdockw{Fixpoint} \coqdef{HilbertIPCsetup.sub}{sub}{\coqdocdefinition{sub}}  (\coqdocvar{s}: \coqexternalref{nat}{http://coq.inria.fr/distrib/8.4pl6/stdlib/Coq.Init.Datatypes}{\coqdocinductive{nat}} \ensuremath{\rightarrow} \coqref{HilbertIPCsetup.form}{\coqdocinductive{form}})  (\coqdocvar{A} : \coqref{HilbertIPCsetup.form}{\coqdocinductive{form}}) : \coqref{HilbertIPCsetup.form}{\coqdocinductive{form}} :=\coqdoceol
\coqdocindent{1.00em}
\coqdockw{match} \coqdocvariable{A} \coqdockw{with}\coqdoceol
\coqdocindent{2.00em}
\ensuremath{|} \coqref{HilbertIPCsetup.var}{\coqdocconstructor{var}} \coqdocvar{i} \ensuremath{\Rightarrow} \coqdocvariable{s} \coqdocvar{i}\coqdoceol
\coqdocindent{2.00em}
\ensuremath{|} \coqdocvar{A} \coqref{HilbertIPCsetup.::x '->>' x}{\coqdocnotation{->>}} \coqdocvar{B} \ensuremath{\Rightarrow} \coqref{HilbertIPCsetup.::x '->>' x}{\coqdocnotation{(}}\coqref{HilbertIPCsetup.sub}{\coqdocdefinition{sub}} \coqdocvariable{s} \coqdocvariable{A}\coqref{HilbertIPCsetup.::x '->>' x}{\coqdocnotation{)}} \coqref{HilbertIPCsetup.::x '->>' x}{\coqdocnotation{->>}} \coqref{HilbertIPCsetup.::x '->>' x}{\coqdocnotation{(}}\coqref{HilbertIPCsetup.sub}{\coqdocdefinition{sub}} \coqdocvariable{s} \coqdocvar{B}\coqref{HilbertIPCsetup.::x '->>' x}{\coqdocnotation{)}}\coqdoceol
\coqdocindent{2.00em}
\ensuremath{|} \coqdocvar{A} \coqref{HilbertIPCsetup.::x 'x26' x}{\coqdocnotation{\&}} \coqdocvar{B} \ensuremath{\Rightarrow} \coqref{HilbertIPCsetup.::x 'x26' x}{\coqdocnotation{(}}\coqref{HilbertIPCsetup.sub}{\coqdocdefinition{sub}} \coqdocvariable{s} \coqdocvariable{A}\coqref{HilbertIPCsetup.::x 'x26' x}{\coqdocnotation{)}} \coqref{HilbertIPCsetup.::x 'x26' x}{\coqdocnotation{\&}} \coqref{HilbertIPCsetup.::x 'x26' x}{\coqdocnotation{(}}\coqref{HilbertIPCsetup.sub}{\coqdocdefinition{sub}} \coqdocvariable{s} \coqdocvar{B}\coqref{HilbertIPCsetup.::x 'x26' x}{\coqdocnotation{)}}\coqdoceol
\coqdocindent{2.00em}
\ensuremath{|} \coqdocvar{A} \coqref{HilbertIPCsetup.::x 'x5Cv/' x}{\coqdocnotation{\symbol{92}}}\coqref{HilbertIPCsetup.::x 'x5Cv/' x}{\coqdocnotation{v}}\coqref{HilbertIPCsetup.::x 'x5Cv/' x}{\coqdocnotation{/}} \coqdocvar{B} \ensuremath{\Rightarrow} \coqref{HilbertIPCsetup.::x 'x5Cv/' x}{\coqdocnotation{(}}\coqref{HilbertIPCsetup.sub}{\coqdocdefinition{sub}} \coqdocvariable{s} \coqdocvariable{A}\coqref{HilbertIPCsetup.::x 'x5Cv/' x}{\coqdocnotation{)}} \coqref{HilbertIPCsetup.::x 'x5Cv/' x}{\coqdocnotation{\symbol{92}}}\coqref{HilbertIPCsetup.::x 'x5Cv/' x}{\coqdocnotation{v}}\coqref{HilbertIPCsetup.::x 'x5Cv/' x}{\coqdocnotation{/}} \coqref{HilbertIPCsetup.::x 'x5Cv/' x}{\coqdocnotation{(}}\coqref{HilbertIPCsetup.sub}{\coqdocdefinition{sub}} \coqdocvariable{s} \coqdocvar{B}\coqref{HilbertIPCsetup.::x 'x5Cv/' x}{\coqdocnotation{)}}\coqdoceol
\coqdocindent{2.00em}
\ensuremath{|} \coqref{HilbertIPCsetup.tt}{\coqdocconstructor{tt}} \ensuremath{\Rightarrow} \coqref{HilbertIPCsetup.tt}{\coqdocconstructor{tt}}\coqdoceol
\coqdocindent{2.00em}
\ensuremath{|} \coqref{HilbertIPCsetup.ff}{\coqdocconstructor{ff}} \ensuremath{\Rightarrow} \coqref{HilbertIPCsetup.ff}{\coqdocconstructor{ff}}\coqdoceol
\coqdocindent{1.00em}
\coqdockw{end}.\coqdoceol
\coqdocemptyline
\coqdocnoindent
\coqdockw{Fixpoint} \coqdef{HilbertIPCsetup.ssub}{ssub}{\coqdocdefinition{ssub}} (\coqdocvar{s} : \coqexternalref{nat}{http://coq.inria.fr/distrib/8.4pl6/stdlib/Coq.Init.Datatypes}{\coqdocinductive{nat}} \ensuremath{\rightarrow} \coqref{HilbertIPCsetup.form}{\coqdocinductive{form}}) (\coqdocvar{G}: \coqref{HilbertIPCsetup.context}{\coqdocabbreviation{context}}) : \coqref{HilbertIPCsetup.context}{\coqdocabbreviation{context}} :=\coqdoceol
\coqdocindent{1.00em}
\coqdockw{match} \coqdocvariable{G} \coqdockw{with}\coqdoceol
\coqdocindent{2.00em}
\ensuremath{|} \coqexternalref{nil}{http://coq.inria.fr/distrib/8.4pl6/stdlib/Coq.Init.Datatypes}{\coqdocconstructor{nil}} \ensuremath{\Rightarrow} \coqexternalref{nil}{http://coq.inria.fr/distrib/8.4pl6/stdlib/Coq.Init.Datatypes}{\coqdocconstructor{nil}}\coqdoceol
\coqdocindent{2.00em}
\ensuremath{|} \coqdocvar{A} \coqexternalref{:list scope:x '::' x}{http://coq.inria.fr/distrib/8.4pl6/stdlib/Coq.Init.Datatypes}{\coqdocnotation{::}} \coqdocvar{G'} \ensuremath{\Rightarrow} \coqexternalref{:list scope:x '::' x}{http://coq.inria.fr/distrib/8.4pl6/stdlib/Coq.Init.Datatypes}{\coqdocnotation{(}}\coqref{HilbertIPCsetup.sub}{\coqdocdefinition{sub}} \coqdocvariable{s} \coqdocvar{A}\coqexternalref{:list scope:x '::' x}{http://coq.inria.fr/distrib/8.4pl6/stdlib/Coq.Init.Datatypes}{\coqdocnotation{)}} \coqexternalref{:list scope:x '::' x}{http://coq.inria.fr/distrib/8.4pl6/stdlib/Coq.Init.Datatypes}{\coqdocnotation{::}} \coqexternalref{:list scope:x '::' x}{http://coq.inria.fr/distrib/8.4pl6/stdlib/Coq.Init.Datatypes}{\coqdocnotation{(}}\coqref{HilbertIPCsetup.ssub}{\coqdocdefinition{ssub}} \coqdocvariable{s} \coqdocvar{G'}\coqexternalref{:list scope:x '::' x}{http://coq.inria.fr/distrib/8.4pl6/stdlib/Coq.Init.Datatypes}{\coqdocnotation{)}}\coqdoceol
\coqdocindent{1.00em}
\coqdockw{end}.\coqdoceol
\coqdocemptyline
\end{coqdoccode}

Typical substitutions arise inductively from a base substitution replacing a single chosen variable by a formula and leaving all other variables unchanged:

\begin{coqdoccode}
\coqdocemptyline
\coqdocnoindent
\coqdockw{Definition} \coqdef{HilbertIPCsetup.s n}{s\_n}{\coqdocdefinition{s\_n}} (\coqdocvar{n}:\coqexternalref{nat}{http://coq.inria.fr/distrib/8.4pl6/stdlib/Coq.Init.Datatypes}{\coqdocinductive{nat}}) (\coqdocvar{A}: \coqref{HilbertIPCsetup.form}{\coqdocinductive{form}}) : (\coqexternalref{nat}{http://coq.inria.fr/distrib/8.4pl6/stdlib/Coq.Init.Datatypes}{\coqdocinductive{nat}} \ensuremath{\rightarrow} \coqref{HilbertIPCsetup.form}{\coqdocinductive{form}}) :=\coqdoceol
\coqdocindent{1.00em}
\coqdockw{fun} \coqdocvar{m}  \ensuremath{\Rightarrow} \coqdockw{match} (\coqexternalref{eq nat dec}{http://coq.inria.fr/distrib/8.4pl6/stdlib/Coq.Arith.Peano\_dec}{\coqdoclemma{eq\_nat\_dec}} \coqdocvariable{n} \coqdocvariable{m}) \coqdockw{with}\coqdoceol
\coqdocindent{6.50em}
\ensuremath{|} \coqexternalref{left}{http://coq.inria.fr/distrib/8.4pl6/stdlib/Coq.Init.Specif}{\coqdocconstructor{left}} \coqdocvar{\_} \ensuremath{\Rightarrow} \coqdocvariable{A}\coqdoceol
\coqdocindent{6.50em}
\ensuremath{|} \coqexternalref{right}{http://coq.inria.fr/distrib/8.4pl6/stdlib/Coq.Init.Specif}{\coqdocconstructor{right}} \coqdocvar{\_} \ensuremath{\Rightarrow} (\coqref{HilbertIPCsetup.var}{\coqdocconstructor{var}} \coqdocvariable{m})\coqdoceol
\coqdocindent{5.50em}
\coqdockw{end}.\coqdoceol
\coqdocemptyline
\coqdocnoindent
\coqdockw{Definition} \coqdef{HilbertIPCsetup.s p}{s\_p}{\coqdocdefinition{s\_p}} := \coqref{HilbertIPCsetup.s n}{\coqdocdefinition{s\_n}} 0.\coqdoceol
\coqdocemptyline
\coqdocnoindent
\coqdockw{Notation} \coqdef{HilbertIPCsetup.::x 'x7B' x 'x7D/' x}{"}{"}A '\{' B '\}/' n " := (\coqref{HilbertIPCsetup.sub}{\coqdocdefinition{sub}} (\coqref{HilbertIPCsetup.s n}{\coqdocdefinition{s\_n}} \coqdocvar{n} \coqdocvar{B}) \coqdocvar{A}) (\coqdoctac{at} \coqdockw{level} 29).\coqdoceol
\coqdocnoindent
\coqdockw{Notation} \coqdef{HilbertIPCsetup.::x 'x7B' x 'x7D/p'}{"}{"}A '\{' B '\}/p'" := (\coqref{HilbertIPCsetup.sub}{\coqdocdefinition{sub}} (\coqref{HilbertIPCsetup.s p}{\coqdocdefinition{s\_p}} \coqdocvar{B}) \coqdocvar{A}) (\coqdoctac{at} \coqdockw{level} 29).\coqdoceol
\coqdocnoindent
\coqdockw{Notation} \coqdef{HilbertIPCsetup.::'ssubp' x x}{"}{"}'ssubp' B G" := (\coqref{HilbertIPCsetup.ssub}{\coqdocdefinition{ssub}} (\coqref{HilbertIPCsetup.s p}{\coqdocdefinition{s\_p}} \coqdocvar{B}) \coqdocvar{G}) (\coqdoctac{at} \coqdockw{level} 29).\coqdoceol
\coqdocemptyline
\end{coqdoccode}

As only the more narrow substitution  \coqdocdefinition{s\_p} for the chosen variable  is needed on most occasions, it does pay off to state suitable lemmas in two versions, one for \coqdocdefinition{s\_n}  and one for \coqdocdefinition{s\_n}, the former usually with postfix \coqdocdefinition{\_gen}. There are also several notions of freshness, for formulas and for lists, both as a predicate and as a boolean-valued recursive function (all distinguished by corresponding suffixes). Finally, we can finalize the notion of iterated substitution:

\begin{coqdoccode}
\coqdocemptyline
\coqdocnoindent
\coqdockw{Fixpoint} \coqdef{HilbertIPCsetup.f p}{f\_p}{\coqdocdefinition{f\_p}} (\coqdocvar{A} : \coqref{HilbertIPCsetup.form}{\coqdocinductive{form}}) (\coqdocvar{n} : \coqexternalref{nat}{http://coq.inria.fr/distrib/8.4pl6/stdlib/Coq.Init.Datatypes}{\coqdocinductive{nat}}) : \coqref{HilbertIPCsetup.form}{\coqdocinductive{form}} :=\coqdoceol
\coqdocindent{1.00em}
\coqdockw{match} \coqdocvariable{n} \coqdockw{with}\coqdoceol
\coqdocindent{2.00em}
\ensuremath{|} 0 \ensuremath{\Rightarrow} \coqref{HilbertIPCsetup.var}{\coqdocconstructor{var}} 0\coqdoceol
\coqdocindent{2.00em}
\ensuremath{|} \coqexternalref{S}{http://coq.inria.fr/distrib/8.4pl6/stdlib/Coq.Init.Datatypes}{\coqdocconstructor{S}} \coqdocvar{n'} \ensuremath{\Rightarrow} \coqref{HilbertIPCsetup.sub}{\coqdocdefinition{sub}} (\coqref{HilbertIPCsetup.s p}{\coqdocdefinition{s\_p}} (\coqref{HilbertIPCsetup.f p}{\coqdocdefinition{f\_p}} \coqdocvariable{A} \coqdocvar{n'})) \coqdocvariable{A}\coqdoceol
\coqdocindent{1.00em}
\coqdockw{end}.\coqdoceol
\coqdocemptyline
\end{coqdoccode}

\begin{remark} \label{rem:extensions}
The definitions so far were fairly straightforward. There are, however, two basic lemmas that are worth singling out and contrasting.

\begin{coqdoccode}
\coqdocemptyline
\coqdocnoindent
\coqdockw{Lemma} \coqdef{HilbertIPCsetup.hil ded}{hil\_ded}{\coqdoclemma{hil\_ded}}: \coqdockw{\ensuremath{\forall}} \coqdocvar{G} (\coqdocvar{A} \coqdocvar{B} : \coqref{HilbertIPCsetup.form}{\coqdocinductive{form}}), \coqdocvariable{A} \coqexternalref{:list scope:x '::' x}{http://coq.inria.fr/distrib/8.4pl6/stdlib/Coq.Init.Datatypes}{\coqdocnotation{::}} \coqdocvariable{G} \coqref{HilbertIPCsetup.::x '|--' x}{\coqdocnotation{|--}} \coqdocvariable{B} \ensuremath{\rightarrow}  \coqdocvariable{G} \coqref{HilbertIPCsetup.::x '|--' x}{\coqdocnotation{|--}} \coqdocvariable{A} \coqref{HilbertIPCsetup.::x '->>' x}{\coqdocnotation{->>}} \coqdocvariable{B}.\coqdoceol
\coqdocnoindent
\coqdockw{Lemma} \coqdef{HilbertIPCsetup.ded subst gen}{ded\_subst\_gen}{\coqdoclemma{ded\_subst\_gen}} : \coqdockw{\ensuremath{\forall}} \coqdocvar{A} \coqdocvar{G} \coqdocvar{B} \coqdocvar{C} \coqdocvar{n}, \coqdocvariable{G} \coqref{HilbertIPCsetup.::x '|--' x}{\coqdocnotation{|--}} \coqdocvariable{B} \coqref{HilbertIPCsetup.::x '<<->>' x}{\coqdocnotation{<<->>}} \coqdocvariable{C} \ensuremath{\rightarrow} 
\mycoqlinebreak
\coqdocvariable{G} \coqref{HilbertIPCsetup.::x '|--' x}{\coqdocnotation{|--}} \coqref{HilbertIPCsetup.::x '<<->>' x}{\coqdocnotation{(}}\coqref{HilbertIPCsetup.sub}{\coqdocdefinition{sub}} (\coqref{HilbertIPCsetup.s n}{\coqdocdefinition{s\_n}} \coqdocvariable{n} \coqdocvariable{B}) \coqdocvariable{A}\coqref{HilbertIPCsetup.::x '<<->>' x}{\coqdocnotation{)}} \coqref{HilbertIPCsetup.::x '<<->>' x}{\coqdocnotation{<<->>}} \coqref{HilbertIPCsetup.::x '<<->>' x}{\coqdocnotation{(}}\coqref{HilbertIPCsetup.sub}{\coqdocdefinition{sub}} (\coqref{HilbertIPCsetup.s n}{\coqdocdefinition{s\_n}} \coqdocvariable{n} \coqdocvariable{C}) \coqdocvariable{A}\coqref{HilbertIPCsetup.::x '<<->>' x}{\coqdocnotation{)}}.\coqdoceol
\coqdocemptyline
\end{coqdoccode} 
\noindent
As a consequence of the last lemma we have

\begin{coqdoccode}
\coqdocemptyline
\coqdocnoindent
\coqdockw{Lemma} \coqdef{HilbertIPCsetup.ded subst}{ded\_subst}{\coqdoclemma{ded\_subst}} : \coqdockw{\ensuremath{\forall}} \coqdocvar{A} \coqdocvar{G} \coqdocvar{B} \coqdocvar{C}, \coqdocvariable{G} \coqref{HilbertIPCsetup.::x '|--' x}{\coqdocnotation{|--}} \coqdocvariable{B} \coqref{HilbertIPCsetup.::x '<<->>' x}{\coqdocnotation{<<->>}} \coqdocvariable{C} \ensuremath{\rightarrow} \mycoqlinebreak
\coqdocvariable{G} \coqref{HilbertIPCsetup.::x '|--' x}{\coqdocnotation{|--}} \coqref{HilbertIPCsetup.::x '<<->>' x}{\coqdocnotation{(}}\coqref{HilbertIPCsetup.sub}{\coqdocdefinition{sub}} (\coqref{HilbertIPCsetup.s p}{\coqdocdefinition{s\_p}} \coqdocvariable{B}) \coqdocvariable{A}\coqref{HilbertIPCsetup.::x '<<->>' x}{\coqdocnotation{)}} \coqref{HilbertIPCsetup.::x '<<->>' x}{\coqdocnotation{<<->>}} \coqref{HilbertIPCsetup.::x '<<->>' x}{\coqdocnotation{(}}\coqref{HilbertIPCsetup.sub}{\coqdocdefinition{sub}} (\coqref{HilbertIPCsetup.s p}{\coqdocdefinition{s\_p}} \coqdocvariable{C}) \coqdocvariable{A}\coqref{HilbertIPCsetup.::x '<<->>' x}{\coqdocnotation{)}}.\coqdoceol
\coqdocemptyline
\end{coqdoccode}

For most non-classical logics, it is by no means common to have both of these metatheorems at the same time. In the modal logic setting, for example, a turnstile relation enjoying this deduction theorem (i.e., the one of the form $\Gamma \cup \{A\}\vdash B$ implies $\Gamma \vdash A \to B$) would be the one known as the \emph{local consequence relation}. However, this notion of consequence does not enjoy the substitution metatheorem. 
 The global consequence relation, which incorporates the Rule of Necessitation, satisfies in turn the latter metatheorem, but not the former one. Similar problems would arise in the realm of substructural logics. Yet both these metatheorems are heavily used in Ruitenburg's work (implicitly) and the present formalization (explicitly), which indicates some reasons why the rarity of periodic sequences property outside the realm of locally finite logics may be not a coincidence. As far as substructural logics are concerned (at least those not enjoying the weakening rule), other examples of incompatible metatheorems heavily used in the present development would include:

\begin{coqdoccode}
\coqdocemptyline
\coqdocnoindent
\coqdockw{Lemma} \coqdef{HilbertIPCsetup.hil weaken}{hil\_weaken}{\coqdoclemma{hil\_weaken}}: \coqdockw{\ensuremath{\forall}} \coqdocvar{G} (\coqdocvar{A} \coqdocvar{B} : \coqref{HilbertIPCsetup.form}{\coqdocinductive{form}}),  \coqdocvariable{G} \coqref{HilbertIPCsetup.::x '|--' x}{\coqdocnotation{|--}} \coqdocvariable{A} \ensuremath{\rightarrow} \coqdocvariable{B} \coqexternalref{:list scope:x '::' x}{http://coq.inria.fr/distrib/8.4pl6/stdlib/Coq.Init.Datatypes}{\coqdocnotation{::}} \coqdocvariable{G} \coqref{HilbertIPCsetup.::x '|--' x}{\coqdocnotation{|--}} \coqdocvariable{A}.\coqdoceol
\coqdocnoindent
\coqdocnoindent
\coqdockw{Lemma} \coqdef{HilbertIPCsetup.hil weaken gen}{hil\_weaken\_gen}{\coqdoclemma{hil\_weaken\_gen}} : \coqdockw{\ensuremath{\forall}} \coqdocvar{G} \coqdocvar{G'} \coqdocvar{A}, \coqdocvariable{G} \coqref{HilbertIPCsetup.::x '|--' x}{\coqdocnotation{|--}} \coqdocvariable{A} \ensuremath{\rightarrow} \coqdocvariable{G'} \coqexternalref{:list scope:x '++' x}{http://coq.inria.fr/distrib/8.4pl6/stdlib/Coq.Init.Datatypes}{\coqdocnotation{++}} \coqdocvariable{G} \coqref{HilbertIPCsetup.::x '|--' x}{\coqdocnotation{|--}} \coqdocvariable{A}.\coqdoceol
\coqdocnoindent
\coqdockw{Lemma} \coqdef{HilbertIPCsetup.hil weaken incl}{hil\_weaken\_incl}{\coqdoclemma{hil\_weaken\_incl}} : \coqdockw{\ensuremath{\forall}} \coqdocvar{G} \coqdocvar{G'} \coqdocvar{A}, 
\coqdocvariable{G} \coqref{HilbertIPCsetup.::x '|--' x}{\coqdocnotation{|--}} \coqdocvariable{A} \ensuremath{\rightarrow} \coqexternalref{incl}{http://coq.inria.fr/distrib/8.4pl6/stdlib/Coq.Lists.List}{\coqdocdefinition{incl}} \coqdocvariable{G} \coqdocvariable{G'} \ensuremath{\rightarrow} \coqdocvariable{G'} \coqref{HilbertIPCsetup.::x '|--' x}{\coqdocnotation{|--}} \coqdocvariable{A}.\coqdoceol\end{coqdoccode}

\medskip

However, \iln{PLL} singled out in \Cref{probl:pll} and more broadly all extensions of \iln{S}provides an important example of a logic enjoying simultaneously all metatheorems listed in the present Remark. Just like for intuitionism itself, there is no gap between its local consequence relation  and its global counterpart. See also Remark \ref{rem:pll} below.
\end{remark}

\subsection{Decidability of the turnstile relation} \label{sec:turnstile}

Ruitenburg's proof of \coqdockw{Theorem} \coqdef{Ruitenburg1984KeyTheorem.rui 1 4}{rui\_1\_4}{\coqdoclemma{rui\_1\_4}} in several places does case analysis, splitting cases between $\Gamma \vdash A$ and $\Gamma \not\vdash A$. While in classical metatheory this does not require further justification, constructively of course it amounts to decidability of the turnstile relation. Simpler syntactic notions such as equality between variables or between  formulas are trivially decidable (and, as one may guess, useful in formal development):

\begin{coqdoccode}
\coqdocemptyline
\coqdocnoindent
\coqdockw{Lemma} \coqdef{HilbertIPCsetup.dceq v}{dceq\_v}{\coqdoclemma{dceq\_v}} : \coqdockw{\ensuremath{\forall}} \coqdocvar{n} \coqdocvar{n0},  \coqexternalref{:type scope:'x7B' x 'x7D' '+' 'x7B' x 'x7D'}{http://coq.inria.fr/distrib/8.4pl6/stdlib/Coq.Init.Specif}{\coqdocnotation{\{}}\coqref{HilbertIPCsetup.var}{\coqdocconstructor{var}} \coqdocvariable{n} \coqexternalref{:type scope:x '=' x}{http://coq.inria.fr/distrib/8.4pl6/stdlib/Coq.Init.Logic}{\coqdocnotation{=}} \coqref{HilbertIPCsetup.var}{\coqdocconstructor{var}} \coqdocvariable{n0}\coqexternalref{:type scope:'x7B' x 'x7D' '+' 'x7B' x 'x7D'}{http://coq.inria.fr/distrib/8.4pl6/stdlib/Coq.Init.Specif}{\coqdocnotation{\}}} \coqexternalref{:type scope:'x7B' x 'x7D' '+' 'x7B' x 'x7D'}{http://coq.inria.fr/distrib/8.4pl6/stdlib/Coq.Init.Specif}{\coqdocnotation{+}} \coqexternalref{:type scope:'x7B' x 'x7D' '+' 'x7B' x 'x7D'}{http://coq.inria.fr/distrib/8.4pl6/stdlib/Coq.Init.Specif}{\coqdocnotation{\{}}\coqref{HilbertIPCsetup.var}{\coqdocconstructor{var}} \coqdocvariable{n} \coqexternalref{:type scope:x '<>' x}{http://coq.inria.fr/distrib/8.4pl6/stdlib/Coq.Init.Logic}{\coqdocnotation{\ensuremath{\not=}}} \coqref{HilbertIPCsetup.var}{\coqdocconstructor{var}} \coqdocvariable{n0}\coqexternalref{:type scope:'x7B' x 'x7D' '+' 'x7B' x 'x7D'}{http://coq.inria.fr/distrib/8.4pl6/stdlib/Coq.Init.Specif}{\coqdocnotation{\}}}.\coqdoceol

\coqdocnoindent
\coqdockw{Lemma} \coqdef{HilbertIPCsetup.dceq f}{dceq\_f}{\coqdoclemma{dceq\_f}}: \coqdockw{\ensuremath{\forall}} (\coqdocvar{A} \coqdocvar{B}: \coqref{HilbertIPCsetup.form}{\coqdocinductive{form}}), \coqexternalref{:type scope:'x7B' x 'x7D' '+' 'x7B' x 'x7D'}{http://coq.inria.fr/distrib/8.4pl6/stdlib/Coq.Init.Specif}{\coqdocnotation{\{}}\coqdocvariable{A} \coqexternalref{:type scope:x '=' x}{http://coq.inria.fr/distrib/8.4pl6/stdlib/Coq.Init.Logic}{\coqdocnotation{=}} \coqdocvariable{B}\coqexternalref{:type scope:'x7B' x 'x7D' '+' 'x7B' x 'x7D'}{http://coq.inria.fr/distrib/8.4pl6/stdlib/Coq.Init.Specif}{\coqdocnotation{\}}} \coqexternalref{:type scope:'x7B' x 'x7D' '+' 'x7B' x 'x7D'}{http://coq.inria.fr/distrib/8.4pl6/stdlib/Coq.Init.Specif}{\coqdocnotation{+}} \coqexternalref{:type scope:'x7B' x 'x7D' '+' 'x7B' x 'x7D'}{http://coq.inria.fr/distrib/8.4pl6/stdlib/Coq.Init.Specif}{\coqdocnotation{\{}}\coqdocvariable{A} \coqexternalref{:type scope:x '<>' x}{http://coq.inria.fr/distrib/8.4pl6/stdlib/Coq.Init.Logic}{\coqdocnotation{\ensuremath{\not=}}} \coqdocvariable{B}\coqexternalref{:type scope:'x7B' x 'x7D' '+' 'x7B' x 'x7D'}{http://coq.inria.fr/distrib/8.4pl6/stdlib/Coq.Init.Specif}{\coqdocnotation{\}}}.\coqdoceol
\coqdocemptyline
\end{coqdoccode}
\noindent
But a constructive proof of decidability of turnstile relation would require incorporating an actual decidability proof for \ipc. As the present development is purely syntactic (and extending it with some standard Kripke completeness proof for \ipc\ would provide  little insight into Ruitenburg's proof), the only route worth considering would be the one mentioned  above: give up the Hilbert-style approach of the original paper \cite{Ruitenburg84:jsl} and use a cut-free sequent system together with an actual proof of cut elimination, then use it to prove decidability of turnstile. This is a viable idea for future development, in particular if extended with proof term assignment to extract the computational content of Ruitenburg's result, especially in the light of more recent work discussed in \Cref{sec:related}. Still, it seems orthogonal to the actual goal of verifying the original proof.  The route taken in the present paper looks as follows:

\begin{coqdoccode}
\coqdocemptyline
\coqdocnoindent
\coqdockw{Module} \coqdef{HilbertIPCsetup.DecidEquiv}{DecidEquiv}{\coqdocmodule{DecidEquiv}}.\coqdoceol
\coqdocemptyline
\coqdocnoindent
\coqdockw{Require} \coqdockw{Import} \coqexternalref{}{http://coq.inria.fr/distrib/8.4pl6/stdlib/Coq.Logic.Classical\_Prop}{\coqdoclibrary{Coq.Logic.Classical\_Prop}}.\coqdoceol
\coqdocemptyline
\coqdocnoindent
\coqdockw{Lemma} \coqdef{HilbertIPCsetup.DecidEquiv.decid equiv}{decid\_equiv}{\coqdoclemma{decid\_equiv}} : \coqdockw{\ensuremath{\forall}} \coqdocvar{G} \coqdocvar{A}, \coqexternalref{:type scope:x 'x5C/' x}{http://coq.inria.fr/distrib/8.4pl6/stdlib/Coq.Init.Logic}{\coqdocnotation{(}}\coqdocvariable{G} \coqref{HilbertIPCsetup.::x '|--' x}{\coqdocnotation{|--}} \coqdocvariable{A}\coqexternalref{:type scope:x 'x5C/' x}{http://coq.inria.fr/distrib/8.4pl6/stdlib/Coq.Init.Logic}{\coqdocnotation{)}} \coqexternalref{:type scope:x 'x5C/' x}{http://coq.inria.fr/distrib/8.4pl6/stdlib/Coq.Init.Logic}{\coqdocnotation{\ensuremath{\lor}}} \coqexternalref{:type scope:'x7E' x}{http://coq.inria.fr/distrib/8.4pl6/stdlib/Coq.Init.Logic}{\coqdocnotation{\~{}(}}\coqdocvariable{G} \coqref{HilbertIPCsetup.::x '|--' x}{\coqdocnotation{|--}} \coqdocvariable{A}\coqexternalref{:type scope:'x7E' x}{http://coq.inria.fr/distrib/8.4pl6/stdlib/Coq.Init.Logic}{\coqdocnotation{)}}.\coqdoceol
\coqdocindent{1.00em}
\coqdoctac{intros}. \coqdoctac{apply} \coqexternalref{classic}{http://coq.inria.fr/distrib/8.4pl6/stdlib/Coq.Logic.Classical\_Prop}{\coqdocaxiom{classic}}.\coqdoceol
\coqdocnoindent
\coqdockw{Qed}.\coqdoceol
\coqdocemptyline
\coqdocnoindent
\coqdockw{End} \coqref{HilbertIPCsetup.DecidEquiv}{\coqdocmodule{DecidEquiv}}.\coqdoceol
\coqdocemptyline
\end{coqdoccode}
\noindent
That is, excluded middle was imported inside a module in order not to contaminate the rest of development (admittedly, these days such import strategy is not encouraged by Coq developers). The reader can verify that it is only used in two places in the proof of \coqdockw{Theorem} \coqdef{Ruitenburg1984KeyTheorem.rui 1 4}{rui\_1\_4}{\coqdoclemma{rui\_1\_4}}.

\section{Proving Ruitenburg's auxiliary lemmas} \label{sec:aux}

So far, we were discussing basic metatheorems used implicitly in Ruitenburg's paper. In this section, we are finally going to start formalizing the actual development in the paper itself, corresponding Lemma 1.2 and Lemmas 1.6--1.8 in the original paper. Lemma 1.2 is the very last part of \cdkw{HilbertIPCSetup}:

\begin{coqdoccode}
\coqdocemptyline
\coqdocnoindent
\coqdockw{Lemma} \coqdef{HilbertIPCsetup.rui 1 2 i}{rui\_1\_2\_i}{\coqdoclemma{rui\_1\_2\_i}} : \coqdockw{\ensuremath{\forall}} \coqdocvar{A} \coqdocvar{i} \coqdocvar{k}, \coqexternalref{ListNotations.:list scope:'[' x ';' '..' ';' x ']'}{http://coq.inria.fr/distrib/8.4pl6/stdlib/Coq.Lists.List}{\coqdocnotation{[}}\coqref{HilbertIPCsetup.f p}{\coqdocdefinition{f\_p}} \coqdocvariable{A} \coqdocvariable{i} \coqexternalref{ListNotations.:list scope:'[' x ';' '..' ';' x ']'}{http://coq.inria.fr/distrib/8.4pl6/stdlib/Coq.Lists.List}{\coqdocnotation{;}} \coqref{HilbertIPCsetup.sub}{\coqdocdefinition{sub}} (\coqref{HilbertIPCsetup.s p}{\coqdocdefinition{s\_p}} \coqref{HilbertIPCsetup.tt}{\coqdocconstructor{tt}}) \coqdocvariable{A}\coqexternalref{ListNotations.:list scope:'[' x ';' '..' ';' x ']'}{http://coq.inria.fr/distrib/8.4pl6/stdlib/Coq.Lists.List}{\coqdocnotation{]}} \coqref{HilbertIPCsetup.::x '|--' x}{\coqdocnotation{|--}} \coqref{HilbertIPCsetup.f p}{\coqdocdefinition{f\_p}} \coqdocvariable{A} (\coqdocvariable{i} \coqexternalref{:nat scope:x '+' x}{http://coq.inria.fr/distrib/8.4pl6/stdlib/Coq.Init.Peano}{\coqdocnotation{+}} \coqdocvariable{k}).\coqdoceol
\coqdocemptyline
\coqdocnoindent
\coqdockw{Lemma} \coqdef{HilbertIPCsetup.rui 1 2 ii}{rui\_1\_2\_ii}{\coqdoclemma{rui\_1\_2\_ii}} : \coqdockw{\ensuremath{\forall}} \coqdocvar{A} \coqdocvar{i} \coqdocvar{n}, \coqexternalref{ListNotations.:list scope:'[' x ';' '..' ';' x ']'}{http://coq.inria.fr/distrib/8.4pl6/stdlib/Coq.Lists.List}{\coqdocnotation{[}} \coqref{HilbertIPCsetup.f p}{\coqdocdefinition{f\_p}} \coqdocvariable{A} \coqdocvariable{i}\coqexternalref{ListNotations.:list scope:'[' x ';' '..' ';' x ']'}{http://coq.inria.fr/distrib/8.4pl6/stdlib/Coq.Lists.List}{\coqdocnotation{;}} \coqref{HilbertIPCsetup.sub}{\coqdocdefinition{sub}} (\coqref{HilbertIPCsetup.s p}{\coqdocdefinition{s\_p}} \coqref{HilbertIPCsetup.tt}{\coqdocconstructor{tt}}) \coqdocvariable{A} \coqexternalref{ListNotations.:list scope:'[' x ';' '..' ';' x ']'}{http://coq.inria.fr/distrib/8.4pl6/stdlib/Coq.Lists.List}{\coqdocnotation{]}} \coqref{HilbertIPCsetup.::x '|--' x}{\coqdocnotation{|--}} 
\mycoqlinebreak
\coqref{HilbertIPCsetup.::x '->>' x}{\coqdocnotation{(}}\coqref{HilbertIPCsetup.f p}{\coqdocdefinition{f\_p}} \coqdocvariable{A} (\coqexternalref{S}{http://coq.inria.fr/distrib/8.4pl6/stdlib/Coq.Init.Datatypes}{\coqdocconstructor{S}} \coqdocvariable{n}) \coqref{HilbertIPCsetup.::x '->>' x}{\coqdocnotation{->>}} \coqref{HilbertIPCsetup.f p}{\coqdocdefinition{f\_p}} \coqdocvariable{A} \coqdocvariable{n}\coqref{HilbertIPCsetup.::x '->>' x}{\coqdocnotation{)}} \coqref{HilbertIPCsetup.::x '->>' x}{\coqdocnotation{->>}} \coqref{HilbertIPCsetup.f p}{\coqdocdefinition{f\_p}} \coqdocvariable{A} \coqdocvariable{n}.\coqdoceol
\coqdocemptyline
\coqdocemptyline
\end{coqdoccode}

Lemmas 1.6--1.8 form the entirety of  \cdkw{Ruitenburg1984Aux}:

\begin{coqdoccode}
\coqdocemptyline
\coqdocemptyline
\coqdocnoindent
\coqdockw{Lemma} \coqdef{Ruitenburg1984Aux.rui 1 6}{rui\_1\_6}{\coqdoclemma{rui\_1\_6}} : \coqdockw{\ensuremath{\forall}} \coqdocvar{A} \coqdocvar{m}, \coqdoceol
\coqdocindent{8.00em}
(\coqdockw{\ensuremath{\forall}} \coqdocvar{i}, \coqexternalref{ListNotations.:list scope:'[' x ';' '..' ';' x ']'}{http://coq.inria.fr/distrib/8.4pl6/stdlib/Coq.Lists.List}{\coqdocnotation{[}}\coqref{HilbertIPCsetup.f p}{\coqdocdefinition{f\_p}} \coqdocvariable{A} \coqdocvariable{i} \coqexternalref{ListNotations.:list scope:'[' x ';' '..' ';' x ']'}{http://coq.inria.fr/distrib/8.4pl6/stdlib/Coq.Lists.List}{\coqdocnotation{;}} \coqref{HilbertIPCsetup.sub}{\coqdocdefinition{sub}} (\coqref{HilbertIPCsetup.s p}{\coqdocdefinition{s\_p}} \coqref{HilbertIPCsetup.tt}{\coqdocconstructor{tt}}) \coqdocvariable{A}\coqexternalref{ListNotations.:list scope:'[' x ';' '..' ';' x ']'}{http://coq.inria.fr/distrib/8.4pl6/stdlib/Coq.Lists.List}{\coqdocnotation{]}} \coqref{HilbertIPCsetup.::x '|--' x}{\coqdocnotation{|--}}
 \coqref{HilbertIPCsetup.f p}{\coqdocdefinition{f\_p}} \coqdocvariable{A} \coqdocvariable{m}) \ensuremath{\rightarrow}
\mycoqlinebreak
\coqexternalref{ListNotations.:list scope:'[' x ';' '..' ';' x ']'}{http://coq.inria.fr/distrib/8.4pl6/stdlib/Coq.Lists.List}{\coqdocnotation{[}}\coqref{HilbertIPCsetup.sub}{\coqdocdefinition{sub}} (\coqref{HilbertIPCsetup.s p}{\coqdocdefinition{s\_p}} \coqref{HilbertIPCsetup.tt}{\coqdocconstructor{tt}}) \coqdocvariable{A}\coqexternalref{ListNotations.:list scope:'[' x ';' '..' ';' x ']'}{http://coq.inria.fr/distrib/8.4pl6/stdlib/Coq.Lists.List}{\coqdocnotation{]}} \coqref{HilbertIPCsetup.::x '|--' x}{\coqdocnotation{|--}} \coqref{HilbertIPCsetup.f p}{\coqdocdefinition{f\_p}} \coqdocvariable{A} \coqdocvariable{m} \coqref{HilbertIPCsetup.::x '<<->>' x}{\coqdocnotation{<<->>}} \coqref{HilbertIPCsetup.f p}{\coqdocdefinition{f\_p}} \coqdocvariable{A} (\coqdocvariable{m}\coqexternalref{:nat scope:x '+' x}{http://coq.inria.fr/distrib/8.4pl6/stdlib/Coq.Init.Peano}{\coqdocnotation{+}}1).\coqdoceol
\coqdocemptyline
\coqdocnoindent
\coqdockw{Lemma} \coqdef{Ruitenburg1984Aux.rui 1 7 i}{rui\_1\_7\_i}{\coqdoclemma{rui\_1\_7\_i}} :  \coqdockw{\ensuremath{\forall}} \coqdocvar{A} \coqdocvar{m} \coqdocvar{n}, \coqdoceol
\coqdocindent{8.00em}
\coqexternalref{ListNotations.:list scope:'[' x ';' '..' ';' x ']'}{http://coq.inria.fr/distrib/8.4pl6/stdlib/Coq.Lists.List}{\coqdocnotation{[}} \coqref{HilbertIPCsetup.sub}{\coqdocdefinition{sub}} (\coqref{HilbertIPCsetup.s p}{\coqdocdefinition{s\_p}} \coqref{HilbertIPCsetup.tt}{\coqdocconstructor{tt}}) (\coqref{HilbertIPCsetup.f p}{\coqdocdefinition{f\_p}} \coqdocvariable{A} (2\coqexternalref{:nat scope:x '*' x}{http://coq.inria.fr/distrib/8.4pl6/stdlib/Coq.Init.Peano}{\coqdocnotation{\ensuremath{\times}}}\coqdocvariable{m} \coqexternalref{:nat scope:x '+' x}{http://coq.inria.fr/distrib/8.4pl6/stdlib/Coq.Init.Peano}{\coqdocnotation{+}} 1))\coqexternalref{ListNotations.:list scope:'[' x ';' '..' ';' x ']'}{http://coq.inria.fr/distrib/8.4pl6/stdlib/Coq.Lists.List}{\coqdocnotation{]}} \coqref{HilbertIPCsetup.::x '|--' x}{\coqdocnotation{|--}}  \coqref{HilbertIPCsetup.sub}{\coqdocdefinition{sub}} (\coqref{HilbertIPCsetup.s p}{\coqdocdefinition{s\_p}} \coqref{HilbertIPCsetup.tt}{\coqdocconstructor{tt}}) (\coqref{HilbertIPCsetup.f p}{\coqdocdefinition{f\_p}} \coqdocvariable{A} \coqdocvariable{n}).\coqdoceol
\coqdocemptyline
\coqdocnoindent
\coqdockw{Lemma} \coqdef{Ruitenburg1984Aux.rui 1 7 ii}{rui\_1\_7\_ii}{\coqdoclemma{rui\_1\_7\_ii}} :  \coqdockw{\ensuremath{\forall}} \coqdocvar{A} \coqdocvar{m} \coqdocvar{n}, \coqdoceol
\coqdocindent{6em}
\coqexternalref{ListNotations.:list scope:'[' x ';' '..' ';' x ']'}{http://coq.inria.fr/distrib/8.4pl6/stdlib/Coq.Lists.List}{\coqdocnotation{[}} \coqref{HilbertIPCsetup.sub}{\coqdocdefinition{sub}} (\coqref{HilbertIPCsetup.s p}{\coqdocdefinition{s\_p}} \coqref{HilbertIPCsetup.tt}{\coqdocconstructor{tt}}) (\coqref{HilbertIPCsetup.f p}{\coqdocdefinition{f\_p}} \coqdocvariable{A} (2\coqexternalref{:nat scope:x '*' x}{http://coq.inria.fr/distrib/8.4pl6/stdlib/Coq.Init.Peano}{\coqdocnotation{\ensuremath{\times}}}\coqdocvariable{m} \coqexternalref{:nat scope:x '+' x}{http://coq.inria.fr/distrib/8.4pl6/stdlib/Coq.Init.Peano}{\coqdocnotation{+}} 2))\coqexternalref{ListNotations.:list scope:'[' x ';' '..' ';' x ']'}{http://coq.inria.fr/distrib/8.4pl6/stdlib/Coq.Lists.List}{\coqdocnotation{]}} \coqref{HilbertIPCsetup.::x '|--' x}{\coqdocnotation{|--}}  \coqref{HilbertIPCsetup.sub}{\coqdocdefinition{sub}} (\coqref{HilbertIPCsetup.s p}{\coqdocdefinition{s\_p}} \coqref{HilbertIPCsetup.tt}{\coqdocconstructor{tt}}) (\coqref{HilbertIPCsetup.f p}{\coqdocdefinition{f\_p}} \coqdocvariable{A} (2 \coqexternalref{:nat scope:x '*' x}{http://coq.inria.fr/distrib/8.4pl6/stdlib/Coq.Init.Peano}{\coqdocnotation{\ensuremath{\times}}}\coqdocvariable{n})).\coqdoceol
\coqdocemptyline
\coqdocnoindent
\coqdockw{Lemma} \coqdef{Ruitenburg1984Aux.rui 1 8}{rui\_1\_8}{\coqdoclemma{rui\_1\_8}} : \coqdockw{\ensuremath{\forall}} \coqdocvar{A} \coqdocvar{m}, \coqdoceol
\coqdocindent{8em} \coqexternalref{ListNotations.:list scope:'[' x ';' '..' ';' x ']'}{http://coq.inria.fr/distrib/8.4pl6/stdlib/Coq.Lists.List}{\coqdocnotation{[}}\coqref{HilbertIPCsetup.sub}{\coqdocdefinition{sub}} (\coqref{HilbertIPCsetup.s p}{\coqdocdefinition{s\_p}} \coqref{HilbertIPCsetup.tt}{\coqdocconstructor{tt}}) \coqdocvariable{A}\coqexternalref{ListNotations.:list scope:'[' x ';' '..' ';' x ']'}{http://coq.inria.fr/distrib/8.4pl6/stdlib/Coq.Lists.List}{\coqdocnotation{]}} \coqref{HilbertIPCsetup.::x '|--' x}{\coqdocnotation{|--}}  \coqref{HilbertIPCsetup.::x '<<->>' x}{\coqdocnotation{(}}\coqref{HilbertIPCsetup.f p}{\coqdocdefinition{f\_p}} \coqdocvariable{A} \coqdocvariable{m}\coqref{HilbertIPCsetup.::x '<<->>' x}{\coqdocnotation{)}}  \coqref{HilbertIPCsetup.::x '<<->>' x}{\coqdocnotation{<<->>}}  \coqref{HilbertIPCsetup.::x '<<->>' x}{\coqdocnotation{(}}\coqref{HilbertIPCsetup.f p}{\coqdocdefinition{f\_p}} \coqdocvariable{A} (\coqdocvariable{m} \coqexternalref{:nat scope:x '+' x}{http://coq.inria.fr/distrib/8.4pl6/stdlib/Coq.Init.Peano}{\coqdocnotation{+}} 1)\coqref{HilbertIPCsetup.::x '<<->>' x}{\coqdocnotation{)}} \ensuremath{\rightarrow}\coqdoceol
\coqdocindent{14.00em}
\coqexternalref{ListNotations.:list scope:'[' ']'}{http://coq.inria.fr/distrib/8.4pl6/stdlib/Coq.Lists.List}{\coqdocnotation{[]}} \coqref{HilbertIPCsetup.::x '|--' x}{\coqdocnotation{|--}} \coqref{HilbertIPCsetup.f p}{\coqdocdefinition{f\_p}} \coqdocvariable{A} (\coqdocvariable{m}\coqexternalref{:nat scope:x '+' x}{http://coq.inria.fr/distrib/8.4pl6/stdlib/Coq.Init.Peano}{\coqdocnotation{+}}1) \coqref{HilbertIPCsetup.::x '<<->>' x}{\coqdocnotation{<<->>}}  \coqref{HilbertIPCsetup.f p}{\coqdocdefinition{f\_p}} \coqdocvariable{A} (\coqdocvariable{m} \coqexternalref{:nat scope:x '+' x}{http://coq.inria.fr/distrib/8.4pl6/stdlib/Coq.Init.Peano}{\coqdocnotation{+}} 3).\coqdoceol
\coqdocnoindent
\end{coqdoccode}

As one can see, these are rather nontrivial metatheorems about \ipc. Still, it was a rather pleasant part of the paper to formalize. 
The automated proofs follow closely proofs in the paper. 
 I would even claim that the Coq proofs are at times easier to understand than in the original version and clarify some remarks that were not always entirely transparent---e.g., the repeated instruction to ``use (iterated) substitution''---but this is ultimately a question of individual preferences.

\begin{remark} \label{rem:pll}
It is worth noting that unlike the main theorem itself (\coqdockw{Theorem} \coqdef{Ruitenburg1984Main.rui 1 9 Ens}{rui\_1\_9\_Ens}{\coqdoclemma{rui\_1\_9\_Ens}}) and its key ingredient (\coqdockw{Theorem} \coqdef{Ruitenburg1984KeyTheorem.rui 1 4}{rui\_1\_4}{\coqdoclemma{rui\_1\_4}}) presented in \Cref{sec:bounds_ens} below, all the lemmas in question would hold for systems meeting the restrictions discussed in Remark \ref{rem:extensions}, in particular to intuitionistic modalities satisfying the axiom \iln{S}, i.e., Curry-Howard counterparts of applicative functors. In fact, the repository includes a fork \cdkw{feature/extended\_formalisms} with a subfolder \cdkw{applicative\_development} illustrating that all the results of \cdkw{HilbertIPCSetup} and  \cdkw{Ruitenburg1984Aux} (i.e., the parts of formalization presented in \Cref{sec:setup} and in this section) would work for any extension $\iln{S}$ in the unimodal Heyting signature. Whether or not there are interesting non-locally-finite logics of this form for which the rest of the development can be carried, i.e., for which llps holds is not clear; cf. Open Problem \ref{probl:pll}.
\end{remark}

\section{Bounds as Ensembles: proving  the main result} \label{sec:bounds_ens}

\coqdockw{Theorem} \coqdef{Ruitenburg1984KeyTheorem.rui 1 4}{rui\_1\_4}{\coqdoclemma{rui\_1\_4}} 
relies on a notion of a \emph{bound} of formula $A$ over a context (a set, or in our case a list of formulas).  It is a finite set of formulas, each of which is equivalent to a substituted implicational subformula of $A$ (more precise definition below), and the proof of \coqdockw{Theorem} \coqdef{Ruitenburg1984KeyTheorem.rui 1 4}{rui\_1\_4}{\coqdoclemma{rui\_1\_4}} proceeds by induction over its cardinality. \coqexternalref{Ensemble}{http://coq.inria.fr/distrib/8.4pl6/stdlib/Coq.Sets.Ensembles}{\coqdocdefinition
{Ensemble}}, an old weapon in Coq's arsenal, seems particularly well-suited for such proofs. The disadvantage, of course, is that being a predicate, i.e., a \coqdockw{Prop}-valued function, it does not allow the use of program extraction or  Coq's programming capabilities; we will see a solution in \Cref{sec:bounds_lists}. I believe, however, that it was beneficial to keep the logical and computational uses of bounds apart. If the code is refined in future, it could be beneficial to explicitly use the notion of \emph{reflection} here.

After giving the obvious definition of \coqdef{BoundsSubformulas.Subformulas}{Subformulas}{\coqdocdefinition{Subformulas}}, we identify their special subclass \coqdef{BoundsSubformulas.BoundSubformulas}{BoundSubformulas}{\coqdocdefinition{BoundSub}}-\coqdef{BoundsSubformulas.BoundSubformulas}{BoundSubformulas}{\coqdocdefinition{formulas}}: those which are either implicational subformulas or propositional variables. Then one can proceed to defining what bounds are:
\begin{coqdoccode}
\coqdocemptyline
\coqdocnoindent
\coqdockw{Definition} \coqdef{BoundsSubformulas.Bound}{Bound}{\coqdocdefinition{Bound}} (\coqdocvar{G}: \coqref{HilbertIPCsetup.context}{\coqdocabbreviation{context}}) (\coqdocvar{A} : \coqref{HilbertIPCsetup.form}{\coqdocinductive{form}}) (\coqdocvar{b} : \coqexternalref{Ensemble}{http://coq.inria.fr/distrib/8.4pl6/stdlib/Coq.Sets.Ensembles}{\coqdocdefinition{Ensemble}} \coqref{HilbertIPCsetup.form}{\coqdocinductive{form}}) :=\coqdoceol
\coqdocnoindent
\coqdocindent{1.00em}
\coqdockw{\ensuremath{\forall}} \coqdocvar{C}: \coqref{HilbertIPCsetup.form}{\coqdocinductive{form}}, \coqexternalref{In}{http://coq.inria.fr/distrib/8.4pl6/stdlib/Coq.Sets.Ensembles}{\coqdocdefinition{In}} (\coqref{BoundsSubformulas.BoundSubformulas}{\coqdocdefinition{BoundSubformulas}} \coqdocvariable{A}) \coqdocvariable{C} \ensuremath{\rightarrow} \coqdoceol
\coqdocindent{8.00em}
\coqexternalref{:type scope:'exists' x '..' x ',' x}{http://coq.inria.fr/distrib/8.4pl6/stdlib/Coq.Init.Logic}{\coqdocnotation{\ensuremath{\exists}}} \coqdocvar{B}\coqexternalref{:type scope:'exists' x '..' x ',' x}{http://coq.inria.fr/distrib/8.4pl6/stdlib/Coq.Init.Logic}{\coqdocnotation{,}} \coqexternalref{In}{http://coq.inria.fr/distrib/8.4pl6/stdlib/Coq.Sets.Ensembles}{\coqdocdefinition{In}} \coqdocvariable{b} \coqdocvariable{B} \coqexternalref{:type scope:x '/x5C' x}{http://coq.inria.fr/distrib/8.4pl6/stdlib/Coq.Init.Logic}{\coqdocnotation{\ensuremath{\land}}} \coqdocvariable{G} \coqref{HilbertIPCsetup.::x '|--' x}{\coqdocnotation{|--}} \coqref{HilbertIPCsetup.sub}{\coqdocdefinition{sub}} (\coqref{HilbertIPCsetup.s p}{\coqdocdefinition{s\_p}} \coqref{HilbertIPCsetup.tt}{\coqdocconstructor{tt}}) \coqdocvariable{C} \coqref{HilbertIPCsetup.::x '<<->>' x}{\coqdocnotation{<<->>}} \coqdocvariable{B}.\coqdoceol
\coqdocemptyline
\end{coqdoccode}

There is a minor discrepancy between the definition of \coqref{BoundsSubformulas.Bound}{\coqdocdefinition{Bound}} used here and the one used in the original paper \cite{Ruitenburg84:jsl}: the latter does not impose that \vrbl{b} already contains formulas (equivalent to) \coqref{BoundsSubformulas.Bound}{\coqdocdefinition{BoundSubformulas}} of \vrbl{A}  elements \coqdocvar{p} substituted with \coqdocvar{tt}; in Ruitenburg's version, the part following the turnstile would be \coqref{HilbertIPCsetup.sub}{\coqdocdefinition{sub}} (\coqref{HilbertIPCsetup.s p}{\coqdocdefinition{s\_p}} \coqref{HilbertIPCsetup.tt}{\coqdocconstructor{tt}}) \coqdocvariable{C} \coqref{HilbertIPCsetup.::x '<<->>' x}{\coqdocnotation{<<->>}} (\coqref{HilbertIPCsetup.s p}{\coqdocdefinition{s\_p}} \coqref{HilbertIPCsetup.tt}{\coqdocconstructor{tt}})  \coqdocvariable{B}.
Of course, if \coqdocvar{b} is a bound in his sense, then \coqdocvar{b}\{\coqdocvar{tt}/\coqdocvar{p}\} is a bound both in the present sense and in his sense;
 hence, the present definition is narrower, but the difference does not matter from a practical point of view. On the other hand, Ruitenburg's definition insist that a bound contains \coqdocconstructor{tt}, whereas it may well happen that a  \coqref{BoundsSubformulas.Bound}{\coqdocdefinition{Bound}} contains nothing equivalent to it; consider, for example, a bound of \vrbl{q} \coqdocnotation{->>} \vrbl{r} over \coqdocnotation{[]}. The difference can be handled easily, as we will see in the statement of \coqdockw{Theorem} \coqdef{Ruitenburg1984KeyTheorem.rui 1 4}{rui\_1\_4}{\coqdoclemma{rui\_1\_4}} below; moreover, Ruitenburg's convention is followed when bounds are treated as lists rather than Ensembles (see \Cref{sec:bounds_lists}). 
 
 \begin{coqdoccode}
\coqdocemptyline
 \coqdocnoindent
\coqdockw{Definition} \coqdef{BoundsSubformulas.ExactBound}{ExactBound}{\coqdocdefinition{ExactBound}} (\coqdocvar{G}: \coqref{HilbertIPCsetup.context}{\coqdocabbreviation{context}}) (\coqdocvar{A} : \coqref{HilbertIPCsetup.form}{\coqdocinductive{form}}) (\coqdocvar{b} : \coqexternalref{Ensemble}{http://coq.inria.fr/distrib/8.4pl6/stdlib/Coq.Sets.Ensembles}{\coqdocdefinition{Ensemble}} \coqref{HilbertIPCsetup.form}{\coqdocinductive{form}}) :=\coqdoceol
\coqdocindent{1.00em}
\coqexternalref{:type scope:x '/x5C' x}{http://coq.inria.fr/distrib/8.4pl6/stdlib/Coq.Init.Logic}{\coqdocnotation{(}}\coqdockw{\ensuremath{\forall}} \coqdocvar{B}: \coqref{HilbertIPCsetup.form}{\coqdocinductive{form}}, \coqexternalref{In}{http://coq.inria.fr/distrib/8.4pl6/stdlib/Coq.Sets.Ensembles}{\coqdocdefinition{In}} \coqdocvariable{b} \coqdocvariable{B} \ensuremath{\rightarrow} \coqdoceol
\coqdocindent{8.00em}
\coqexternalref{:type scope:'exists' x '..' x ',' x}{http://coq.inria.fr/distrib/8.4pl6/stdlib/Coq.Init.Logic}{\coqdocnotation{\ensuremath{\exists}}} \coqdocvar{C}\coqexternalref{:type scope:'exists' x '..' x ',' x}{http://coq.inria.fr/distrib/8.4pl6/stdlib/Coq.Init.Logic}{\coqdocnotation{,}} \coqexternalref{In}{http://coq.inria.fr/distrib/8.4pl6/stdlib/Coq.Sets.Ensembles}{\coqdocdefinition{In}} (\coqref{BoundsSubformulas.BoundSubformulas}{\coqdocdefinition{BoundSubformulas}} \coqdocvariable{A}) \coqdocvariable{C} \coqexternalref{:type scope:x '/x5C' x}{http://coq.inria.fr/distrib/8.4pl6/stdlib/Coq.Init.Logic}{\coqdocnotation{\ensuremath{\land}}} \coqdoceol
\coqdocindent{10.30em}
\coqdocvariable{G} \coqref{HilbertIPCsetup.::x '|--' x}{\coqdocnotation{|--}} \coqref{HilbertIPCsetup.sub}{\coqdocdefinition{sub}} (\coqref{HilbertIPCsetup.s p}{\coqdocdefinition{s\_p}} \coqref{HilbertIPCsetup.tt}{\coqdocconstructor{tt}}) \coqdocvariable{C} \coqref{HilbertIPCsetup.::x '<<->>' x}{\coqdocnotation{<<->>}} \coqdocvariable{B}\coqexternalref{:type scope:x '/x5C' x}{http://coq.inria.fr/distrib/8.4pl6/stdlib/Coq.Init.Logic}{\coqdocnotation{)}} \coqexternalref{:type scope:x '/x5C' x}{http://coq.inria.fr/distrib/8.4pl6/stdlib/Coq.Init.Logic}{\coqdocnotation{\ensuremath{\land}}} \coqdoceol
\coqdocindent{1.00em}
\coqref{BoundsSubformulas.Bound}{\coqdocdefinition{Bound}} \coqdocvariable{G} \coqdocvariable{A} \coqdocvariable{b}\coqdoceol
\coqdocemptyline
\end{coqdoccode}
\noindent
 The remainder of \cdkw{BoundsSubformulas.v} is devoted to various auxiliary lemmas. 

We can move on now to the contents of \coqdockw{Ruitenburg1984KeyTheorem}, which contains the actual proof of the central syntactic result in the paper. 
The first larger theorem proved in this file, before actual \coqdockw{Theorem} \coqdef{Ruitenburg1984KeyTheorem.rui 1 4}{rui\_1\_4}{\coqdoclemma{rui\_1\_4}}, deals with a remark in the base case of the proof, referring to the Rieger-Nishimura theorem mentioned in \Cref{sec:intro}. Recall that this very theorem pinpoints why \ipc\ does not have local finiteness by describing the infinite poset of all \ipc-formulas in one free variable, quotiented by provable equivalence and ordered by provable implication. While the Rieger-Nishimura lattice has been thoroughly understood  and reconstructed on several occasions, using differing techniques\footnote{In fact, its very name refers to the rediscovery of Rieger's result \cite{Rieger49} by Nishimura \cite{Nishimura60:jsl}. More information and further references can be found in standard monographs (see, e.g., \cite[Ch. 7]{ChagrovZ97:ml}).}, it could be of interest to formalize it fully. Fortunately, as it turns out, we only need a corollary of this result:

\begin{coqdoccode}
\coqdocemptyline
\coqdocnoindent
\coqdockw{Lemma}  \coqdef{Ruitenburg1984KeyTheorem.Rieger Nishimura corollary}{Rieger\_Nishimura\_corollary}{\coqdoclemma{Rieger\_Nishimura\_corollary}} :\coqdockw{\ensuremath{\forall}} \coqdocvar{G} \coqdocvar{B} \coqdocvar{v},  \coqdoceol
\coqdocindent{5.50em}
\coqref{HilbertIPCsetup.fresh l p}{\coqdocdefinition{fresh\_l\_p}} \coqdocvariable{v} (\coqdocvariable{B} \coqexternalref{:list scope:x '::' x}{http://coq.inria.fr/distrib/8.4pl6/stdlib/Coq.Init.Datatypes}{\coqdocnotation{::}} \coqdocvariable{G}) \ensuremath{\rightarrow}  \coqdoceol
\coqdocindent{5.50em}
( \coqdockw{\ensuremath{\forall}} \coqdocvar{C} : \coqref{HilbertIPCsetup.form}{\coqdocinductive{form}},
(\coqref{BoundsSubformulas.BoundSubformulas}{\coqdocdefinition{BoundSubformulas}} \coqdocvariable{B}) \coqdocvariable{C} \ensuremath{\rightarrow}
\coqdocvariable{G} \coqref{HilbertIPCsetup.::x '|--' x}{\coqdocnotation{|--}} \coqref{HilbertIPCsetup.sub}{\coqdocdefinition{sub}} (\coqref{HilbertIPCsetup.s p}{\coqdocdefinition{s\_p}} \coqref{HilbertIPCsetup.tt}{\coqdocconstructor{tt}}) \coqdocvariable{C}  ) \ensuremath{\rightarrow}
\coqdockw{\ensuremath{\forall}}  \coqdocvar{C},  (\coqref{BoundsSubformulas.Subformulas}{\coqdocdefinition{Subformulas}} \coqdocvariable{B}) \coqdocvariable{C} \ensuremath{\rightarrow}  \coqdocvariable{G} \coqref{HilbertIPCsetup.::x '|--' x}{\coqdocnotation{|--}} \coqref{HilbertIPCsetup.::x '->>' x}{\coqdocnotation{(}}\coqref{HilbertIPCsetup.sub}{\coqdocdefinition{sub}} (\coqref{HilbertIPCsetup.s p}{\coqdocdefinition{s\_p}} (\coqref{HilbertIPCsetup.var}{\coqdocconstructor{var}} \coqdocvariable{v})) \coqdocvariable{C}\coqref{HilbertIPCsetup.::x '->>' x}{\coqdocnotation{)}} \coqref{HilbertIPCsetup.::x '->>' x}{\coqdocnotation{->>}} \coqref{HilbertIPCsetup.ff}{\coqdocconstructor{ff}} \coqexternalref{:type scope:x 'x5C/' x}{http://coq.inria.fr/distrib/8.4pl6/stdlib/Coq.Init.Logic}{\coqdocnotation{\ensuremath{\lor}}}  \coqdoceol
\coqdocindent{12.00em}
\coqdocvariable{G} \coqref{HilbertIPCsetup.::x '|--' x}{\coqdocnotation{|--}} \coqref{HilbertIPCsetup.sub}{\coqdocdefinition{sub}} (\coqref{HilbertIPCsetup.s p}{\coqdocdefinition{s\_p}} (\coqref{HilbertIPCsetup.var}{\coqdocconstructor{var}} \coqdocvariable{v})) \coqdocvariable{C} \coqref{HilbertIPCsetup.::x '<<->>' x}{\coqdocnotation{<<->>}} \coqref{HilbertIPCsetup.::x '<<->>' x}{\coqdocnotation{(}}\coqref{HilbertIPCsetup.var}{\coqdocconstructor{var}} \coqdocvariable{v}\coqref{HilbertIPCsetup.::x '<<->>' x}{\coqdocnotation{)}}  \coqexternalref{:type scope:x 'x5C/' x}{http://coq.inria.fr/distrib/8.4pl6/stdlib/Coq.Init.Logic}{\coqdocnotation{\ensuremath{\lor}}}\coqdoceol
\coqdocindent{12.00em}
\coqdocvariable{G} \coqref{HilbertIPCsetup.::x '|--' x}{\coqdocnotation{|--}} \coqref{HilbertIPCsetup.sub}{\coqdocdefinition{sub}} (\coqref{HilbertIPCsetup.s p}{\coqdocdefinition{s\_p}} (\coqref{HilbertIPCsetup.var}{\coqdocconstructor{var}} \coqdocvariable{v})) \coqdocvariable{C}.\coqdoceol
\coqdocemptyline
\end{coqdoccode}

\noindent
With this last issue out of the way, one can finally state and prove

\newcommand{\smallerint}{\coqdocindent{1.00em}}
\newcommand{\biggerint}{\coqdocindent{2.50em}}
\newcommand{\biggestint}{\coqdocindent{3.20em}}

 \begin{coqdoccode}
\coqdocemptyline
\coqdocnoindent
\coqdockw{Theorem} \coqdef{Ruitenburg1984KeyTheorem.rui 1 4}{rui\_1\_4}{\coqdoclemma{rui\_1\_4}}: \coqdockw{\ensuremath{\forall}} \coqdocvar{n} \coqdocvar{G} \coqdocvar{A} \coqdocvar{B},\coqdoceol
\smallerint
\coqexternalref{Included}{http://coq.inria.fr/distrib/8.4pl6/stdlib/Coq.Sets.Ensembles}{\coqdocdefinition{Included}} (\coqref{BoundsSubformulas.BoundSubformulas}{\coqdocdefinition{BoundSubformulas}} \coqdocvariable{B}) (\coqref{BoundsSubformulas.BoundSubformulas}{\coqdocdefinition{BoundSubformulas}} \coqdocvariable{A}) \ensuremath{\rightarrow} \coqdoceol
\smallerint
\coqdockw{\ensuremath{\forall}} \coqdocvar{i} \coqdocvar{v},\coqdoceol
\smallerint
\coqdockw{let} \coqdocvar{G'} :=  \coqexternalref{:list scope:x '::' x}{http://coq.inria.fr/distrib/8.4pl6/stdlib/Coq.Init.Datatypes}{\coqdocnotation{(}}\coqref{HilbertIPCsetup.f p}{\coqdocdefinition{f\_p}} \coqdocvariable{A} \coqdocvariable{i}\coqexternalref{:list scope:x '::' x}{http://coq.inria.fr/distrib/8.4pl6/stdlib/Coq.Init.Datatypes}{\coqdocnotation{)}} \coqexternalref{:list scope:x '::' x}{http://coq.inria.fr/distrib/8.4pl6/stdlib/Coq.Init.Datatypes}{\coqdocnotation{::}} \coqref{HilbertIPCsetup.sub}{\coqdocdefinition{sub}} (\coqref{HilbertIPCsetup.s p}{\coqdocdefinition{s\_p}} \coqref{HilbertIPCsetup.tt}{\coqdocconstructor{tt}}) \coqdocvariable{A} \coqexternalref{:list scope:x '::' x}{http://coq.inria.fr/distrib/8.4pl6/stdlib/Coq.Init.Datatypes}{\coqdocnotation{::}} \coqdocvariable{G} \coqdoctac{in}\coqdoceol
\smallerint
(\coqexternalref{:type scope:'exists' x '..' x ',' x}{http://coq.inria.fr/distrib/8.4pl6/stdlib/Coq.Init.Logic}{\coqdocnotation{\ensuremath{\exists}}} \coqdocvar{b}\coqexternalref{:type scope:'exists' x '..' x ',' x}{http://coq.inria.fr/distrib/8.4pl6/stdlib/Coq.Init.Logic}{\coqdocnotation{,}}  
\coqdockw{let} \coqdocvar{b'} := (\coqref{BoundsSubformulas.App}{\coqdocdefinition{App}} \coqref{HilbertIPCsetup.tt}{\coqdocconstructor{tt}} \coqdocvariable{b}) \coqdoctac{in} 
\coqref{BoundsSubformulas.Bound}{\coqdocdefinition{Bound}} \coqdocvariable{G'} \coqdocvariable{A} \coqdocvariable{b'} \coqexternalref{:type scope:x '/x5C' x}{http://coq.inria.fr/distrib/8.4pl6/stdlib/Coq.Init.Logic}{\coqdocnotation{\ensuremath{\land}}} \coqexternalref{cardinal}{http://coq.inria.fr/distrib/8.4pl6/stdlib/Coq.Sets.Finite\_sets}{\coqdocinductive{cardinal}} \coqdocvariable{b} \coqdocvariable{n}) \ensuremath{\rightarrow}\coqdoceol
\smallerint
\coqref{HilbertIPCsetup.fresh l p}{\coqdocdefinition{fresh\_l\_p}} \coqdocvariable{v} (\coqref{HilbertIPCsetup.p}{\coqdocdefinition{p}}\coqexternalref{:list scope:x '::' x}{http://coq.inria.fr/distrib/8.4pl6/stdlib/Coq.Init.Datatypes}{\coqdocnotation{::}}\coqdocvariable{B}\coqexternalref{:list scope:x '::' x}{http://coq.inria.fr/distrib/8.4pl6/stdlib/Coq.Init.Datatypes}{\coqdocnotation{::}}\coqdocvariable{A}\coqexternalref{:list scope:x '::' x}{http://coq.inria.fr/distrib/8.4pl6/stdlib/Coq.Init.Datatypes}{\coqdocnotation{::}}\coqdocvariable{G}) \ensuremath{\rightarrow}\coqdoceol
\biggerint
\coqexternalref{:list scope:x '::' x}{http://coq.inria.fr/distrib/8.4pl6/stdlib/Coq.Init.Datatypes}{\coqdocnotation{(}}\coqref{HilbertIPCsetup.::x '->>' x}{\coqdocnotation{(}}\coqref{HilbertIPCsetup.f p}{\coqdocdefinition{f\_p}} \coqdocvariable{A} (2\coqexternalref{:nat scope:x '*' x}{http://coq.inria.fr/distrib/8.4pl6/stdlib/Coq.Init.Peano}{\coqdocnotation{\ensuremath{\times}}}\coqdocvariable{n})\coqref{HilbertIPCsetup.::x '->>' x}{\coqdocnotation{)}} \coqref{HilbertIPCsetup.::x '->>' x}{\coqdocnotation{->>}} \coqref{HilbertIPCsetup.::x '->>' x}{\coqdocnotation{(}}\coqref{HilbertIPCsetup.var}{\coqdocconstructor{var}} \coqdocvariable{v}\coqref{HilbertIPCsetup.::x '->>' x}{\coqdocnotation{)}}\coqexternalref{:list scope:x '::' x}{http://coq.inria.fr/distrib/8.4pl6/stdlib/Coq.Init.Datatypes}{\coqdocnotation{)::}}\coqdocvariable{G'} \coqref{HilbertIPCsetup.::x '|--' x}{\coqdocnotation{|--}} \coqdoceol
\biggestint
\coqref{HilbertIPCsetup.::x 'x26' x}{\coqdocnotation{(}}\coqref{HilbertIPCsetup.sub}{\coqdocdefinition{sub}} (\coqref{HilbertIPCsetup.s p}{\coqdocdefinition{s\_p}} (\coqref{HilbertIPCsetup.var}{\coqdocconstructor{var}} \coqdocvariable{v})) \coqdocvariable{B} \coqref{HilbertIPCsetup.::x '<<->>' x}{\coqdocnotation{<<->>}} \coqref{HilbertIPCsetup.sub}{\coqdocdefinition{sub}} (\coqref{HilbertIPCsetup.s p}{\coqdocdefinition{s\_p}} \coqref{HilbertIPCsetup.tt}{\coqdocconstructor{tt}}) \coqdocvariable{B}\coqref{HilbertIPCsetup.::x 'x26' x}{\coqdocnotation{)}} \coqref{HilbertIPCsetup.::x 'x26' x}{\coqdocnotation{\&}} \coqref{HilbertIPCsetup.::x 'x26' x}{\coqdocnotation{(}}\coqref{HilbertIPCsetup.sub}{\coqdocdefinition{sub}} (\coqref{HilbertIPCsetup.s p}{\coqdocdefinition{s\_p}} \coqref{HilbertIPCsetup.tt}{\coqdocconstructor{tt}}) \coqdocvariable{B} \coqref{HilbertIPCsetup.::x '->>' x}{\coqdocnotation{->>}} \coqref{HilbertIPCsetup.::x '->>' x}{\coqdocnotation{(}}\coqref{HilbertIPCsetup.var}{\coqdocconstructor{var}} \coqdocvariable{v}\coqref{HilbertIPCsetup.::x '->>' x}{\coqdocnotation{)}}\coqref{HilbertIPCsetup.::x 'x26' x}{\coqdocnotation{)}} \coqexternalref{:type scope:x 'x5C/' x}{http://coq.inria.fr/distrib/8.4pl6/stdlib/Coq.Init.Logic}{\coqdocnotation{\ensuremath{\lor}}}\coqdoceol
\biggerint
\coqexternalref{:list scope:x '::' x}{http://coq.inria.fr/distrib/8.4pl6/stdlib/Coq.Init.Datatypes}{\coqdocnotation{(}}\coqref{HilbertIPCsetup.::x '->>' x}{\coqdocnotation{(}}\coqref{HilbertIPCsetup.f p}{\coqdocdefinition{f\_p}} \coqdocvariable{A} (2\coqexternalref{:nat scope:x '*' x}{http://coq.inria.fr/distrib/8.4pl6/stdlib/Coq.Init.Peano}{\coqdocnotation{\ensuremath{\times}}}\coqdocvariable{n})\coqref{HilbertIPCsetup.::x '->>' x}{\coqdocnotation{)}} \coqref{HilbertIPCsetup.::x '->>' x}{\coqdocnotation{->>}} \coqref{HilbertIPCsetup.::x '->>' x}{\coqdocnotation{(}}\coqref{HilbertIPCsetup.var}{\coqdocconstructor{var}} \coqdocvariable{v}\coqref{HilbertIPCsetup.::x '->>' x}{\coqdocnotation{)}}\coqexternalref{:list scope:x '::' x}{http://coq.inria.fr/distrib/8.4pl6/stdlib/Coq.Init.Datatypes}{\coqdocnotation{)::}}\coqdocvariable{G'} \coqref{HilbertIPCsetup.::x '|--' x}{\coqdocnotation{|--}} 
\coqref{HilbertIPCsetup.sub}{\coqdocdefinition{sub}} (\coqref{HilbertIPCsetup.s p}{\coqdocdefinition{s\_p}} (\coqref{HilbertIPCsetup.var}{\coqdocconstructor{var}} \coqdocvariable{v})) \coqdocvariable{B} \coqref{HilbertIPCsetup.::x '<<->>' x}{\coqdocnotation{<<->>}}  \coqref{HilbertIPCsetup.::x '<<->>' x}{\coqdocnotation{(}}\coqref{HilbertIPCsetup.var}{\coqdocconstructor{var}} \coqdocvariable{v}\coqref{HilbertIPCsetup.::x '<<->>' x}{\coqdocnotation{)}} \coqexternalref{:type scope:x 'x5C/' x}{http://coq.inria.fr/distrib/8.4pl6/stdlib/Coq.Init.Logic}{\coqdocnotation{\ensuremath{\lor}}}\coqdoceol
\biggerint
\coqexternalref{:list scope:x '::' x}{http://coq.inria.fr/distrib/8.4pl6/stdlib/Coq.Init.Datatypes}{\coqdocnotation{(}}\coqref{HilbertIPCsetup.::x '->>' x}{\coqdocnotation{(}}\coqref{HilbertIPCsetup.f p}{\coqdocdefinition{f\_p}} \coqdocvariable{A} (2\coqexternalref{:nat scope:x '*' x}{http://coq.inria.fr/distrib/8.4pl6/stdlib/Coq.Init.Peano}{\coqdocnotation{\ensuremath{\times}}}\coqdocvariable{n})\coqref{HilbertIPCsetup.::x '->>' x}{\coqdocnotation{)}} \coqref{HilbertIPCsetup.::x '->>' x}{\coqdocnotation{->>}} \coqref{HilbertIPCsetup.::x '->>' x}{\coqdocnotation{(}}\coqref{HilbertIPCsetup.var}{\coqdocconstructor{var}} \coqdocvariable{v}\coqref{HilbertIPCsetup.::x '->>' x}{\coqdocnotation{)}}\coqexternalref{:list scope:x '::' x}{http://coq.inria.fr/distrib/8.4pl6/stdlib/Coq.Init.Datatypes}{\coqdocnotation{)::}}\coqdocvariable{G'} \coqref{HilbertIPCsetup.::x '|--' x}{\coqdocnotation{|--}} 
\coqref{HilbertIPCsetup.sub}{\coqdocdefinition{sub}} (\coqref{HilbertIPCsetup.s p}{\coqdocdefinition{s\_p}} (\coqref{HilbertIPCsetup.var}{\coqdocconstructor{var}} \coqdocvariable{v})) \coqdocvariable{B}.\coqdoceol
\coqdocnoindent
\coqdocemptyline
\end{coqdoccode}

A comparison reveals inessential differences with the statement of this theorem in the original paper. Instead of assuming that \coqdocdefinition{BoundSubformulas} of \vrbl{B} are contained in those of \vrbl{A}, a premise of the result stated by Ruitenburg is the existence of a bound (of size \vrbl{n}) of \vrbl{A} \& \vrbl{B} over \vrbl{G'} .  Then, at the very beginning of the proof, an observation is made that one can assume that \emph{\vrbl{B} is a subformula of \vrbl{A} by replacing \vrbl{A} by the equivalent formula} \vrbl{A}~\&~(\vrbl{B}~$\backslash$~v/~\coqdocconstructor{tt}). 
The assumption made here, i.e., that
  (\coqref{BoundsSubformulas.BoundSubformulas}{\coqdocdefinition{BoundSubformulas}} \coqdocvariable{B}) are included in (\coqref{BoundsSubformulas.BoundSubformulas}{\coqdocdefinition{BoundSubformulas}} \coqdocvariable{A}) 
seems an optimal solution. 
 Another, very minor difference is that the supposed bound \vrbl{b} is immediately tweaked to \vrbl{b'} containing \coqdocconstructor{tt}. This has to do with the difference in the definition of our \coqdocdefinition{Bound} mentioned above.

The proof is by induction on \vrbl{n}, i.e., the  cardinality of 
a bound for \vrbl{A} (and, consequently, also for \vrbl{B}). The inductive step, furthermore, involves an induction over \vrbl{B}. This in turn involves a consideration of numerous cases and subcases of these subcases, almost each of which involves  some actual propositional reasoning in \ipc. 


The theorem yields corollaries grouped in the file \coqdockw{Ruitenburg1984Main.v}:

\begin{coqdoccode}
\coqdocemptyline
\coqdocnoindent
\coqdockw{Corollary} \coqdef{Ruitenburg1984Main.rui 1 4'}{rui\_1\_4'}{\coqdoclemma{rui\_1\_4'}}: \coqdockw{\ensuremath{\forall}} \coqdocvar{A} \coqdocvar{b} \coqdocvar{n}, (\coqref{BoundsSubformulas.Bound}{\coqdocdefinition{Bound}}  \coqexternalref{ListNotations.:list scope:'[' ']'}{http://coq.inria.fr/distrib/8.4pl6/stdlib/Coq.Lists.List}{\coqdocnotation{[]}} \coqdocvariable{A} \coqdocvariable{b}) \ensuremath{\rightarrow} \coqexternalref{cardinal}{http://coq.inria.fr/distrib/8.4pl6/stdlib/Coq.Sets.Finite\_sets}{\coqdocinductive{cardinal}} \coqdocvariable{b} \coqdocvariable{n} \ensuremath{\rightarrow}\coqdoceol
\smallerint
\coqdockw{\ensuremath{\forall}} \coqdocvar{i} \coqdocvar{v}, 
\coqdockw{let} \coqdocvar{v'} := \coqexternalref{S}{http://coq.inria.fr/distrib/8.4pl6/stdlib/Coq.Init.Datatypes}{\coqdocconstructor{S}} \coqdocvariable{v} \coqdoctac{in} 
\coqref{HilbertIPCsetup.fresh f p}{\coqdocinductive{fresh\_f\_p}} \coqdocvariable{v'} \coqdocvariable{A} \ensuremath{\rightarrow}\coqdoceol
\biggerint
\coqdockw{let} \coqdocvar{G'} :=  \coqexternalref{ListNotations.:list scope:'[' x ';' '..' ';' x ']'}{http://coq.inria.fr/distrib/8.4pl6/stdlib/Coq.Lists.List}{\coqdocnotation{[}}\coqref{HilbertIPCsetup.f p}{\coqdocdefinition{f\_p}} \coqdocvariable{A} \coqdocvariable{i} \coqexternalref{ListNotations.:list scope:'[' x ';' '..' ';' x ']'}{http://coq.inria.fr/distrib/8.4pl6/stdlib/Coq.Lists.List}{\coqdocnotation{;}} \coqref{HilbertIPCsetup.sub}{\coqdocdefinition{sub}} (\coqref{HilbertIPCsetup.s p}{\coqdocdefinition{s\_p}} \coqref{HilbertIPCsetup.tt}{\coqdocconstructor{tt}}) \coqdocvariable{A}\coqexternalref{ListNotations.:list scope:'[' x ';' '..' ';' x ']'}{http://coq.inria.fr/distrib/8.4pl6/stdlib/Coq.Lists.List}{\coqdocnotation{]}}  \coqdoctac{in}\coqdoceol
\biggerint
\coqdockw{let} \coqdocvar{G'{}'} :=  \coqexternalref{:list scope:x '::' x}{http://coq.inria.fr/distrib/8.4pl6/stdlib/Coq.Init.Datatypes}{\coqdocnotation{(}}\coqref{HilbertIPCsetup.f p}{\coqdocdefinition{f\_p}} \coqdocvariable{A} (2\coqexternalref{:nat scope:x '*' x}{http://coq.inria.fr/distrib/8.4pl6/stdlib/Coq.Init.Peano}{\coqdocnotation{\ensuremath{\times}}}\coqdocvariable{n}) \coqref{HilbertIPCsetup.::x '->>' x}{\coqdocnotation{->>}} \coqref{HilbertIPCsetup.::x '->>' x}{\coqdocnotation{(}}\coqref{HilbertIPCsetup.var}{\coqdocconstructor{var}} \coqdocvariable{v'}\coqref{HilbertIPCsetup.::x '->>' x}{\coqdocnotation{)}}\coqexternalref{:list scope:x '::' x}{http://coq.inria.fr/distrib/8.4pl6/stdlib/Coq.Init.Datatypes}{\coqdocnotation{)::}}\coqdocvariable{G'}  \coqdoctac{in}\coqdoceol
\biggestint
\coqdocvariable{G'{}'} \coqref{HilbertIPCsetup.::x '|--' x}{\coqdocnotation{|--}} \coqref{HilbertIPCsetup.::x 'x26' x}{\coqdocnotation{(}}\coqref{HilbertIPCsetup.sub}{\coqdocdefinition{sub}} (\coqref{HilbertIPCsetup.s p}{\coqdocdefinition{s\_p}} (\coqref{HilbertIPCsetup.var}{\coqdocconstructor{var}} \coqdocvariable{v'})) \coqdocvariable{A} \coqref{HilbertIPCsetup.::x '<<->>' x}{\coqdocnotation{<<->>}} \coqref{HilbertIPCsetup.sub}{\coqdocdefinition{sub}} (\coqref{HilbertIPCsetup.s p}{\coqdocdefinition{s\_p}} \coqref{HilbertIPCsetup.tt}{\coqdocconstructor{tt}}) \coqdocvariable{A}\coqref{HilbertIPCsetup.::x 'x26' x}{\coqdocnotation{)}} \coqref{HilbertIPCsetup.::x 'x26' x}{\coqdocnotation{\&}} \coqdoceol
\biggestint\coqdocindent{3em}
\coqref{HilbertIPCsetup.::x 'x26' x}{\coqdocnotation{(}}\coqref{HilbertIPCsetup.sub}{\coqdocdefinition{sub}} (\coqref{HilbertIPCsetup.s p}{\coqdocdefinition{s\_p}} \coqref{HilbertIPCsetup.tt}{\coqdocconstructor{tt}}) \coqdocvariable{A} \coqref{HilbertIPCsetup.::x '->>' x}{\coqdocnotation{->>}} \coqref{HilbertIPCsetup.::x '->>' x}{\coqdocnotation{(}}\coqref{HilbertIPCsetup.var}{\coqdocconstructor{var}} \coqdocvariable{v'}\coqref{HilbertIPCsetup.::x '->>' x}{\coqdocnotation{)}}\coqref{HilbertIPCsetup.::x 'x26' x}{\coqdocnotation{)}} \coqexternalref{:type scope:x 'x5C/' x}{http://coq.inria.fr/distrib/8.4pl6/stdlib/Coq.Init.Logic}{\coqdocnotation{\ensuremath{\lor}}}\coqdoceol
\biggestint
\coqdocvariable{G'{}'} \coqref{HilbertIPCsetup.::x '|--' x}{\coqdocnotation{|--}} \coqref{HilbertIPCsetup.sub}{\coqdocdefinition{sub}} (\coqref{HilbertIPCsetup.s p}{\coqdocdefinition{s\_p}} (\coqref{HilbertIPCsetup.var}{\coqdocconstructor{var}} \coqdocvariable{v'})) \coqdocvariable{A} \coqref{HilbertIPCsetup.::x '<<->>' x}{\coqdocnotation{<<->>}}  \coqref{HilbertIPCsetup.::x '<<->>' x}{\coqdocnotation{(}}\coqref{HilbertIPCsetup.var}{\coqdocconstructor{var}} \coqdocvariable{v'}\coqref{HilbertIPCsetup.::x '<<->>' x}{\coqdocnotation{)}} \coqexternalref{:type scope:x 'x5C/' x}{http://coq.inria.fr/distrib/8.4pl6/stdlib/Coq.Init.Logic}{\coqdocnotation{\ensuremath{\lor}}}\coqdoceol
\biggestint
\coqdocvariable{G'{}'} \coqref{HilbertIPCsetup.::x '|--' x}{\coqdocnotation{|--}} \coqref{HilbertIPCsetup.sub}{\coqdocdefinition{sub}} (\coqref{HilbertIPCsetup.s p}{\coqdocdefinition{s\_p}} (\coqref{HilbertIPCsetup.var}{\coqdocconstructor{var}} \coqdocvariable{v'})) \coqdocvariable{A}.\coqdoceol
\coqdocemptyline
\end{coqdoccode}

In other words, this corollary is simply stating what \coqdef{Ruitenburg1984KeyTheorem.rui 1 4}{rui\_1\_4}{\coqdoclemma{rui\_1\_4}} means for a single formula \vrbl{A} rather than for a pair of formulas \vrbl{A}, \vrbl{B} (note one needs to ensure here that the fresh variable in question is not \vrbl{p} itself; it might happen that \vrbl{A} does not contain it).

\begin{coqdoccode}
\coqdocemptyline
\coqdocnoindent
\coqdockw{Corollary} \coqdef{Ruitenburg1984Main.rui 1 5}{rui\_1\_5}{\coqdoclemma{rui\_1\_5}}: \coqdockw{\ensuremath{\forall}} \coqdocvar{A} \coqdocvar{b} \coqdocvar{i} \coqdocvar{n}, (\coqref{BoundsSubformulas.Bound}{\coqdocdefinition{Bound}} \coqexternalref{ListNotations.:list scope:'[' ']'}{http://coq.inria.fr/distrib/8.4pl6/stdlib/Coq.Lists.List}{\coqdocnotation{[]}} \coqdocvariable{A} \coqdocvariable{b}) \ensuremath{\rightarrow} \coqexternalref{cardinal}{http://coq.inria.fr/distrib/8.4pl6/stdlib/Coq.Sets.Finite\_sets}{\coqdocinductive{cardinal}} \coqdocvariable{b} \coqdocvariable{n} \ensuremath{\rightarrow} \coqdoceol
\biggerint
 \coqdockw{let} \coqdocvar{m}:= (2 \coqexternalref{:nat scope:x '*' x}{http://coq.inria.fr/distrib/8.4pl6/stdlib/Coq.Init.Peano}{\coqdocnotation{\ensuremath{\times}}} \coqdocvariable{n} \coqexternalref{:nat scope:x '+' x}{http://coq.inria.fr/distrib/8.4pl6/stdlib/Coq.Init.Peano}{\coqdocnotation{+}} 1) \coqdoctac{in} \coqdoceol
 \biggestint
  \coqexternalref{ListNotations.:list scope:'[' x ';' '..' ';' x ']'}{http://coq.inria.fr/distrib/8.4pl6/stdlib/Coq.Lists.List}{\coqdocnotation{[}}\coqref{HilbertIPCsetup.sub}{\coqdocdefinition{sub}} (\coqref{HilbertIPCsetup.s p}{\coqdocdefinition{s\_p}} \coqref{HilbertIPCsetup.tt}{\coqdocconstructor{tt}}) \coqdocvariable{A}\coqexternalref{ListNotations.:list scope:'[' x ';' '..' ';' x ']'}{http://coq.inria.fr/distrib/8.4pl6/stdlib/Coq.Lists.List}{\coqdocnotation{;}} \coqref{HilbertIPCsetup.f p}{\coqdocdefinition{f\_p}} \coqdocvariable{A} \coqdocvariable{i}\coqexternalref{ListNotations.:list scope:'[' x ';' '..' ';' x ']'}{http://coq.inria.fr/distrib/8.4pl6/stdlib/Coq.Lists.List}{\coqdocnotation{]}} \coqref{HilbertIPCsetup.::x '|--' x}{\coqdocnotation{|--}} \coqref{HilbertIPCsetup.f p}{\coqdocdefinition{f\_p}} \coqdocvariable{A} \coqdocvariable{m}.\coqdoceol
\coqdocemptyline
\end{coqdoccode}

While it may seem surprising, proving this much simpler-looking corollary is the only goal of the intimidating \coqdockw{Theorem} \coqdef{Ruitenburg1984KeyTheorem.rui 1 4}{rui\_1\_4}{\coqdoclemma{rui\_1\_4}}, and the proof  of this corollary is not even using the theorem itself, but rather \coqdockw{Corollary} \coqdef{Ruitenburg1984Main.rui 1 4'}{rui\_1\_4'}{\coqdoclemma{rui\_1\_4'}} above. Note also that in the original paper, the statement of the theorem does not involve the connection between \vrbl{m} and a bound, although it is clear in the proof. 

This corollary is now combined with Lemmas 1.6--1.8 discussed in \Cref{sec:aux} to yield 

\begin{coqdoccode}
\coqdocemptyline
\coqdocnoindent
\coqdockw{Theorem} \coqdef{Ruitenburg1984Main.rui 1 9 Ens}{rui\_1\_9\_Ens}{\coqdoclemma{rui\_1\_9\_Ens}}:\coqdoceol
\coqdocindent{1.00em}
\coqdockw{\ensuremath{\forall}} \coqdocvar{A} \coqdocvar{b} \coqdocvar{m}, (\coqref{BoundsSubformulas.Bound}{\coqdocdefinition{Bound}} \coqexternalref{ListNotations.:list scope:'[' ']'}{http://coq.inria.fr/distrib/8.4pl6/stdlib/Coq.Lists.List}{\coqdocnotation{[]}} \coqdocvariable{A} \coqdocvariable{b}) \ensuremath{\rightarrow}\coqdoceol
\coqdocindent{8.00em}
\coqexternalref{cardinal}{http://coq.inria.fr/distrib/8.4pl6/stdlib/Coq.Sets.Finite\_sets}{\coqdocinductive{cardinal}} \coqdocvariable{b} \coqdocvariable{m} \ensuremath{\rightarrow}\coqdoceol
\coqdocindent{8.00em}
\coqexternalref{ListNotations.:list scope:'[' ']'}{http://coq.inria.fr/distrib/8.4pl6/stdlib/Coq.Lists.List}{\coqdocnotation{[]}} \coqref{HilbertIPCsetup.::x '|--' x}{\coqdocnotation{|--}} \coqref{HilbertIPCsetup.f p}{\coqdocdefinition{f\_p}} \coqdocvariable{A} (2 \coqexternalref{:nat scope:x '*' x}{http://coq.inria.fr/distrib/8.4pl6/stdlib/Coq.Init.Peano}{\coqdocnotation{\ensuremath{\times}}} \coqdocvariable{m} \coqexternalref{:nat scope:x '+' x}{http://coq.inria.fr/distrib/8.4pl6/stdlib/Coq.Init.Peano}{\coqdocnotation{+}} 2) \coqref{HilbertIPCsetup.::x '<<->>' x}{\coqdocnotation{<<->>}}  \coqref{HilbertIPCsetup.f p}{\coqdocdefinition{f\_p}} \coqdocvariable{A} (2 \coqexternalref{:nat scope:x '*' x}{http://coq.inria.fr/distrib/8.4pl6/stdlib/Coq.Init.Peano}{\coqdocnotation{\ensuremath{\times}}} \coqdocvariable{m} \coqexternalref{:nat scope:x '+' x}{http://coq.inria.fr/distrib/8.4pl6/stdlib/Coq.Init.Peano}{\coqdocnotation{+}} 4).\coqdoceol
\coqdocemptyline
\end{coqdoccode}

This finally explains the name \coqref{BoundsSubformulas.Bound}{\coqdocdefinition{Bound}}: the size of such a bound for \vrbl{A} determines (linearly) how many iterated substitutions (at worst) it takes before it enters the cycle. Clearly, there is always a bound linear in the size of \vrbl{A}: simply take all the implicational subformulas (i.e., \coqdef{BoundsSubformulas.BoundSubformulas}{BoundSubformulas}{\coqdocdefinition{BoundSubformulas}}) of \vrbl{A} and substitute \coqdocconstructor{tt} for \vrbl{p} removing possible duplicates. Is it the best that one can do? And is it important to care about such minor adjustments? This is the last part of our considerations.

\section{Bounds as lists: computations and program extraction} \label{sec:bounds_lists}
 
 As discussed in \Cref{sec:bounds_ens}, treatment of bounds as Ensembles, very convenient for the proof of main result, is not very useful computationally, starting from the fact that \coqdockw{Prop} gets erased during program extraction. For this reason, one can consider a natural reformulation of the notion of a bound in terms of lists. These, in turn, would be awkward in proofs  discussed in \Cref{sec:bounds_ens}, but are  bread-and-butter from a functional programming point of view. All that is needed to verify the programs obtained in this way is to provide bridge theorems between Ensemble- and list-counterparts of the same notion, and this is much easier than developing everything from scratch in terms of lists. These are the contents of \coqdockw{BoundsLists.v}. As suggested above, a more structured approach would probably involve systematic use of reflection.
 
A suitable counterpart of \coqref{BoundsSubformulas.Bound}{\coqdocdefinition{Bound}}  is

\begin{coqdoccode}
\coqdocemptyline
\coqdocnoindent
\coqdockw{Definition} \coqdef{BoundsLists.bound}{bound}{\coqdocdefinition{bound}} (\coqdocvar{b}: \coqexternalref{list}{http://coq.inria.fr/distrib/8.4pl6/stdlib/Coq.Init.Datatypes}{\coqdocinductive{list}} \coqref{HilbertIPCsetup.form}{\coqdocinductive{form}}) (\coqdocvar{A} : \coqref{HilbertIPCsetup.form}{\coqdocinductive{form}}) (\coqdocvar{G} : \coqref{HilbertIPCsetup.context}{\coqdocabbreviation{context}}) :=\coqdoceol
\coqdocindent{1.00em}
\coqexternalref{Forall}{http://coq.inria.fr/distrib/8.4pl6/stdlib/Coq.Lists.List}{\coqdocinductive{Forall}} (\coqdockw{fun} \coqdocvar{C} \ensuremath{\Rightarrow} \coqexternalref{Exists}{http://coq.inria.fr/distrib/8.4pl6/stdlib/Coq.Lists.List}{\coqdocinductive{Exists}} (\coqdockw{fun} \coqdocvar{B} \ensuremath{\Rightarrow} \coqdocvariable{G} \coqref{HilbertIPCsetup.::x '|--' x}{\coqdocnotation{|--}}  \coqdocvariable{B} \coqref{HilbertIPCsetup.::x '<<->>' x}{\coqdocnotation{<<->>}}  \coqdocvariable{C} \coqexternalref{:type scope:x '/x5C' x}{http://coq.inria.fr/distrib/8.4pl6/stdlib/Coq.Init.Logic}{\coqdocnotation{\ensuremath{\land}}} \coqexternalref{ListNotations.:list scope:'[' ']'}{http://coq.inria.fr/distrib/8.4pl6/stdlib/Coq.Lists.List}{\coqdocnotation{[]}} \coqref{HilbertIPCsetup.::x '|--' x}{\coqdocnotation{|--}} \coqdocvariable{B} \coqref{HilbertIPCsetup.::x '<<->>' x}{\coqdocnotation{<<->>}} \coqref{HilbertIPCsetup.sub}{\coqdocdefinition{sub}} (\coqref{HilbertIPCsetup.s p}{\coqdocdefinition{s\_p}} \coqref{HilbertIPCsetup.tt}{\coqdocconstructor{tt}}) \coqdocvariable{B}) \coqdocvariable{b}) (\coqref{BoundsLists.mb red}{\coqdocdefinition{mb\_red}} \coqdocvariable{A}).\coqdoceol
\coqdocemptyline
\coqdocemptyline
\end{coqdoccode}
\noindent
Using a natural function converting contexts (i.e., lists) to Ensembles, one can easily show

\begin{coqdoccode}
\coqdocemptyline
\coqdocnoindent
\coqdockw{Lemma} \coqdef{BoundsLists.bound is Bound}{bound\_is\_Bound}{\coqdoclemma{bound\_is\_Bound}} : \coqdockw{\ensuremath{\forall}} \coqdocvar{b} \coqdocvar{A} \coqdocvar{G}, \coqref{BoundsLists.bound}{\coqdocdefinition{bound}} \coqdocvariable{b} \coqdocvariable{A} \coqdocvariable{G} \ensuremath{\rightarrow} \coqref{BoundsSubformulas.Bound}{\coqdocdefinition{Bound}} \coqdocvariable{G} \coqdocvariable{A} (\coqref{BoundsSubformulas.context to set}{\coqdocdefinition{context\_to\_set}} \coqdocvariable{b}).\coqdoceol
\coqdocemptyline
\end{coqdoccode}
\noindent
And the corresponding version of the main theorem is

\begin{coqdoccode}
\coqdocemptyline
\coqdocnoindent
\coqdockw{Theorem} \coqdef{Ruitenburg1984Main.rui 1 9 list}{rui\_1\_9\_list}{\coqdoclemma{rui\_1\_9\_list}}: \coqdockw{\ensuremath{\forall}} \coqdocvar{A} \coqdocvar{b}, (\coqref{BoundsLists.bound}{\coqdocdefinition{bound}} \coqdocvariable{b} \coqdocvariable{A} \coqexternalref{ListNotations.:list scope:'[' ']'}{http://coq.inria.fr/distrib/8.4pl6/stdlib/Coq.Lists.List}{\coqdocnotation{[]}}) \ensuremath{\rightarrow} \coqdoceol
\smallerint
 \coqexternalref{:type scope:'exists' x '..' x ',' x}{http://coq.inria.fr/distrib/8.4pl6/stdlib/Coq.Init.Logic}{\coqdocnotation{\ensuremath{\exists}}} \coqdocvar{m}\coqexternalref{:type scope:'exists' x '..' x ',' x}{http://coq.inria.fr/distrib/8.4pl6/stdlib/Coq.Init.Logic}{\coqdocnotation{,}} \coqdocvariable{m} \coqexternalref{:nat scope:x '<=' x}{http://coq.inria.fr/distrib/8.4pl6/stdlib/Coq.Init.Peano}{\coqdocnotation{\ensuremath{\le}}} \coqexternalref{length}{http://coq.inria.fr/distrib/8.4pl6/stdlib/Coq.Init.Datatypes}{\coqdocdefinition{length}} \coqdocvariable{b} \coqexternalref{:type scope:x '/x5C' x}{http://coq.inria.fr/distrib/8.4pl6/stdlib/Coq.Init.Logic}{\coqdocnotation{\ensuremath{\land}}}
\coqexternalref{ListNotations.:list scope:'[' ']'}{http://coq.inria.fr/distrib/8.4pl6/stdlib/Coq.Lists.List}{\coqdocnotation{[]}} \coqref{HilbertIPCsetup.::x '|--' x}{\coqdocnotation{|--}} \coqref{HilbertIPCsetup.f p}{\coqdocdefinition{f\_p}} \coqdocvariable{A} (2 \coqexternalref{:nat scope:x '*' x}{http://coq.inria.fr/distrib/8.4pl6/stdlib/Coq.Init.Peano}{\coqdocnotation{\ensuremath{\times}}} \coqdocvariable{m} \coqexternalref{:nat scope:x '+' x}{http://coq.inria.fr/distrib/8.4pl6/stdlib/Coq.Init.Peano}{\coqdocnotation{+}} 2) \coqref{HilbertIPCsetup.::x '<<->>' x}{\coqdocnotation{<<->>}}  \coqref{HilbertIPCsetup.f p}{\coqdocdefinition{f\_p}} \coqdocvariable{A} (2 \coqexternalref{:nat scope:x '*' x}{http://coq.inria.fr/distrib/8.4pl6/stdlib/Coq.Init.Peano}{\coqdocnotation{\ensuremath{\times}}} \coqdocvariable{m} \coqexternalref{:nat scope:x '+' x}{http://coq.inria.fr/distrib/8.4pl6/stdlib/Coq.Init.Peano}{\coqdocnotation{+}} 4).\coqdoceol
\coqdocnoindent
\coqdocemptyline
\end{coqdoccode}

The most na\"ive way to produce a \coqdef{BoundsLists.bound}{bound}{\coqdocdefinition{bound}} for \vrbl{A} over \coqexternalref{ListNotations.:list scope:'[' ']'}{http://coq.inria.fr/distrib/8.4pl6/stdlib/Coq.Lists.List}{\coqdocnotation{[]}} (and hence over any \vrbl{G}, as implied by \coqdef{BoundsLists.bound for bound upward}{bound\_for\_bound\_upward}{\coqdoclemma{bound\_for\_bound\_upward}}) is by 

\begin{coqdoccode}
\coqdocemptyline
\coqdocnoindent
\coqdockw{Fixpoint} \coqdef{BoundsLists.mb red}{mb\_red}{\coqdocdefinition{mb\_red}} (\coqdocvar{A} : \coqref{HilbertIPCsetup.form}{\coqdocinductive{form}}) : \coqexternalref{list}{http://coq.inria.fr/distrib/8.4pl6/stdlib/Coq.Init.Datatypes}{\coqdocinductive{list}} \coqref{HilbertIPCsetup.form}{\coqdocinductive{form}} :=\coqdoceol
\coqdocnoindent
\coqdockw{match} \coqdocvariable{A} \coqdockw{with}\coqdoceol
\coqdocindent{2.00em}
\ensuremath{|} \coqref{HilbertIPCsetup.var}{\coqdocconstructor{var}} \coqdocvar{i} \ensuremath{\Rightarrow} \coqexternalref{ListNotations.:list scope:'[' x ';' '..' ';' x ']'}{http://coq.inria.fr/distrib/8.4pl6/stdlib/Coq.Lists.List}{\coqdocnotation{[}}\coqref{HilbertIPCsetup.sub}{\coqdocdefinition{sub}} (\coqref{HilbertIPCsetup.s p}{\coqdocdefinition{s\_p}} \coqref{HilbertIPCsetup.tt}{\coqdocconstructor{tt}}) (\coqref{HilbertIPCsetup.var}{\coqdocconstructor{var}} \coqdocvar{i}) \coqexternalref{ListNotations.:list scope:'[' x ';' '..' ';' x ']'}{http://coq.inria.fr/distrib/8.4pl6/stdlib/Coq.Lists.List}{\coqdocnotation{;}} \coqref{HilbertIPCsetup.tt}{\coqdocconstructor{tt}}\coqexternalref{ListNotations.:list scope:'[' x ';' '..' ';' x ']'}{http://coq.inria.fr/distrib/8.4pl6/stdlib/Coq.Lists.List}{\coqdocnotation{]}}\coqdoceol
\coqdocindent{2.00em}
\ensuremath{|} \coqdocvar{B} \coqref{HilbertIPCsetup.::x '->>' x}{\coqdocnotation{->>}} \coqdocvar{C} \ensuremath{\Rightarrow} \coqref{HilbertIPCsetup.sub}{\coqdocdefinition{sub}} (\coqref{HilbertIPCsetup.s p}{\coqdocdefinition{s\_p}} \coqref{HilbertIPCsetup.tt}{\coqdocconstructor{tt}}) (\coqdocvar{B} \coqref{HilbertIPCsetup.::x '->>' x}{\coqdocnotation{->>}} \coqdocvar{C}) \coqexternalref{:list scope:x '::' x}{http://coq.inria.fr/distrib/8.4pl6/stdlib/Coq.Init.Datatypes}{\coqdocnotation{::}} \coqexternalref{:list scope:x '::' x}{http://coq.inria.fr/distrib/8.4pl6/stdlib/Coq.Init.Datatypes}{\coqdocnotation{(}}\coqref{BoundsLists.mb red}{\coqdocdefinition{mb\_red}} \coqdocvar{B} \coqexternalref{:list scope:x '++' x}{http://coq.inria.fr/distrib/8.4pl6/stdlib/Coq.Init.Datatypes}{\coqdocnotation{++}} \coqref{BoundsLists.mb red}{\coqdocdefinition{mb\_red}} \coqdocvar{C}\coqexternalref{:list scope:x '::' x}{http://coq.inria.fr/distrib/8.4pl6/stdlib/Coq.Init.Datatypes}{\coqdocnotation{)}}\coqdoceol
\coqdocindent{2.00em}
\ensuremath{|} \coqdocvar{B} \coqref{HilbertIPCsetup.::x 'x26' x}{\coqdocnotation{\&}} \coqdocvar{C} \ensuremath{\Rightarrow} (\coqref{BoundsLists.mb red}{\coqdocdefinition{mb\_red}} \coqdocvar{B} \coqexternalref{:list scope:x '++' x}{http://coq.inria.fr/distrib/8.4pl6/stdlib/Coq.Init.Datatypes}{\coqdocnotation{++}} \coqref{BoundsLists.mb red}{\coqdocdefinition{mb\_red}} \coqdocvar{C})\coqdoceol
\coqdocindent{2.00em}
\ensuremath{|} \coqdocvar{B} \coqref{HilbertIPCsetup.::x 'x5Cv/' x}{\coqdocnotation{\symbol{92}}}\coqref{HilbertIPCsetup.::x 'x5Cv/' x}{\coqdocnotation{v}}\coqref{HilbertIPCsetup.::x 'x5Cv/' x}{\coqdocnotation{/}} \coqdocvar{C} \ensuremath{\Rightarrow} (\coqref{BoundsLists.mb red}{\coqdocdefinition{mb\_red}} \coqdocvar{B} \coqexternalref{:list scope:x '++' x}{http://coq.inria.fr/distrib/8.4pl6/stdlib/Coq.Init.Datatypes}{\coqdocnotation{++}} \coqref{BoundsLists.mb red}{\coqdocdefinition{mb\_red}} \coqdocvar{C})\coqdoceol
\coqdocindent{2.00em}
\ensuremath{|} \coqref{HilbertIPCsetup.tt}{\coqdocconstructor{tt}} \ensuremath{\Rightarrow} \coqexternalref{ListNotations.:list scope:'[' x ';' '..' ';' x ']'}{http://coq.inria.fr/distrib/8.4pl6/stdlib/Coq.Lists.List}{\coqdocnotation{[}}\coqref{HilbertIPCsetup.tt}{\coqdocconstructor{tt}}\coqexternalref{ListNotations.:list scope:'[' x ';' '..' ';' x ']'}{http://coq.inria.fr/distrib/8.4pl6/stdlib/Coq.Lists.List}{\coqdocnotation{]}}\coqdoceol
\coqdocindent{2.00em}
\ensuremath{|} \coqref{HilbertIPCsetup.ff}{\coqdocconstructor{ff}} \ensuremath{\Rightarrow} \coqexternalref{ListNotations.:list scope:'[' x ';' '..' ';' x ']'}{http://coq.inria.fr/distrib/8.4pl6/stdlib/Coq.Lists.List}{\coqdocnotation{[}}\coqref{HilbertIPCsetup.tt}{\coqdocconstructor{tt}}\coqexternalref{ListNotations.:list scope:'[' x ';' '..' ';' x ']'}{http://coq.inria.fr/distrib/8.4pl6/stdlib/Coq.Lists.List}{\coqdocnotation{]}}\coqdoceol
\coqdocnoindent
\coqdockw{end}.\coqdoceol
\coqdocemptyline
\end{coqdoccode}
\noindent
Note that this time we are following Ruitenburg's convention and explicitly including \coqref{HilbertIPCsetup.tt}{\coqdocconstructor{tt}}. 

However, this is obviously suboptimal. To begin with, the output of \coqdef{BoundsLists.mb red}{mb\_red}{\coqdocdefinition{mb\_red}} is almost guaranteed to contain duplicates, but this is easy to deal with (using \coqdef{BoundsLists.dup rem}{dup\_rem}{\coqdocdefinition{dup\_rem}}). More importantly, such a list is also likely to contain equivalent formulas, which are also redundant. The problem gets dramatic when the formula in question contains no other variables than \vrbl{p}; cf. Proposition 2.3 and Theorem 2.4 in the original paper \cite{Ruitenburg84:jsl};  within the one-variable fragment, \ipc\ has strictly globally periodic sequences, just like the classical logic. But improvements are possible also when a formula contains more than one variable. The present development is restricted to a simple optimizer \coqdef{BoundsLists.t optimize}{t\_optimize}{\coqdocdefinition{t\_optimize}}, which is essentially removing redundant occurrences of \coqref{HilbertIPCsetup.tt}{\coqdocconstructor{tt}}.  The function \coqdef{BoundsLists.optimized bound}{optimized\_bound}{\coqdocdefinition{optimized\_bound}} combines duplicate removal from \coqdef{BoundsLists.dup rem}{dup\_rem}{\coqdocdefinition{dup\_rem}} with iterating  \coqdef{BoundsLists.t optimize}{t\_optimize}{\coqdocdefinition{t\_optimize}} as many times as the formula depth of \vrbl{A} requires. Still further improvements are conceivable. Given that \ipc\ is decidable, the ultimate option would be to integrate a decision procedure for \ipc\ (cf. \Cref{sec:turnstile}) and test pairwise elements of a given \coqdef{BoundsLists.bound}{bound}{\coqdocdefinition{bound}}, removing the elements equivalent to those  found earlier in the list. 

Nevertheless, the present stage of development can already be used for actual computation, either via Coq's core functional programming language Gallina or, if one prefers, via program extraction. Combo functions available at the end of   \coqdockw{Ruitenburg1984Main.v}, i.e., \coqdocvar{optimized\_cycle}  or \coqdocvar{cycle\_formula\_length} produce the value after which the sequence is going to enter a cycle for a given $A(p)$ and the size of  corresponding $A^{2m + 2}(p)$. They can be directly extracted to any typical target language such as Haskell or OCaml. In fact, Coq's \coqdockw{Compute} itself does a satisfying job in computing these values. However, even simple experiments indicate one should be rather careful as a blow-up can occur very quickly. The fact that $m$ itself is linear in the size of $A$ surely enough does not mean that $A^{2m + 2}(p)$ is and rather simple examples can make it painfully clear.  Consider, e.g., a formula from \coqdockw{BoundsLists.v}:

\begin{coqdoccode}
\coqdocemptyline
\coqdocnoindent
\coqdockw{Definition} \coqdef{BoundsLists.exform1}{exform1}{\coqdocdefinition{exform1}} := \coqref{HilbertIPCsetup.::x 'x26' x}{\coqdocnotation{(}}\coqref{HilbertIPCsetup.q}{\coqdocdefinition{q}} \coqref{HilbertIPCsetup.::x '->>' x}{\coqdocnotation{->>}} \coqref{HilbertIPCsetup.p}{\coqdocdefinition{p}} \coqref{HilbertIPCsetup.::x '->>' x}{\coqdocnotation{->>}}\coqref{HilbertIPCsetup.r}{\coqdocdefinition{r}}\coqref{HilbertIPCsetup.::x 'x26' x}{\coqdocnotation{)}} \coqref{HilbertIPCsetup.::x 'x26' x}{\coqdocnotation{\&}} \coqref{HilbertIPCsetup.::x 'x26' x}{\coqdocnotation{(}}\coqref{HilbertIPCsetup.::x '->>' x}{\coqdocnotation{(}}\coqref{HilbertIPCsetup.p}{\coqdocdefinition{p}} \coqref{HilbertIPCsetup.::x '->>' x}{\coqdocnotation{->>}} \coqref{HilbertIPCsetup.r}{\coqdocdefinition{r}}\coqref{HilbertIPCsetup.::x '->>' x}{\coqdocnotation{)}} \coqref{HilbertIPCsetup.::x '->>' x}{\coqdocnotation{->>}} \coqref{HilbertIPCsetup.p}{\coqdocdefinition{p}} \coqref{HilbertIPCsetup.::x 'x5Cv/' x}{\coqdocnotation{\symbol{92}}}\coqref{HilbertIPCsetup.::x 'x5Cv/' x}{\coqdocnotation{v}}\coqref{HilbertIPCsetup.::x 'x5Cv/' x}{\coqdocnotation{/}} \coqref{HilbertIPCsetup.r}{\coqdocdefinition{r}}\coqref{HilbertIPCsetup.::x 'x26' x}{\coqdocnotation{)}}.\coqdoceol
\coqdocemptyline
\end{coqdoccode}
\noindent
for which the length of the value of \coqdef{BoundsLists.optimized bound}g{optimized\_bound}{\coqdocdefinition{optimized\_bound}} is just $m := 4$.  The reader is now encouraged to pick $A(p)$ to be \coqdef{BoundsLists.exform1}{exform1}{\coqdocdefinition{exform1}}  and, as an exercise, estimate the size of $A^{2m+2}(p)$.\footnote{This is, in fact, a mistake I made myself while experimenting with the code. I ran Gallina on this input without a prior pen-and-paper or simply commonsensical estimate of the size of output. To Gallina's credit,  it was able to return with the actual formula after several minutes. A curious reader can see it in \coqdockw{f\_p\_exform1\_10\_output.txt}.} 

 \section{Conclusions} 
 \label{sec:conclusions}

As the formalization was developed in 2015--2016, and it was an exercise for the author in understanding Ruitenburg's paper and improving his own skills, it does not involve most modern or complex Coq libraries and features. After relatively minor changes, it has proved possible to compile under recent versions of Coq (upwards of 8.17), but the development itself is not using in an essential way anything that was not already available in versions  8.4pl6 and 8.5. Some libraries being used are already getting obsolete, but a proper overhaul would constitute a separate project, focusing directly on the theorem-proving community. \Cref{sec:related} and the work of F{\'{e}}r{\'{e}}e and van Gool \cite{FereeG23} or Shillito and coauthors \cite{ShillitoGGI23,ShillitoG22,GoreRS21} suggest how such an overhaul could potentially look like. 

The routes for future development have been already suggested in the paper. I find the question whether there are other natural non-locally-finite logics with llps particularly intriguing (Open Problems \ref{probl:pll} and \ref{probl:relevance}). Combining the present formalization with some standard proof of decidability of \ipc \ and using it, e.g., to eliminate altogether the meta-level Excluded Middle (\Cref{sec:turnstile}) or to compute optimal size of a bound for any input (\Cref{sec:bounds_lists}) also seems a natural challenge for future work. 

\label{sect:bib}

\bibliographystyle{eptcs}
\bibliography{fics2024ruit.bib}

\end{document}